\documentclass[11pt, a4paper]{article}

\usepackage{amsmath}
\usepackage{amssymb}
\usepackage{amsfonts}
\usepackage{amsthm}
\usepackage{graphicx}
\usepackage{fullpage}
\usepackage{thm-restate}
\usepackage{color}
\usepackage{booktabs}
\usepackage{hyperref}

\usepackage{cite}

\newtheorem{theorem}{Theorem}[section]
\newtheorem{lemma}[theorem]{Lemma}
\newtheorem{proposition}[theorem]{Proposition}
\newtheorem{corollary}[theorem]{Corollary}
\newtheorem{claim}[theorem]{Claim}
\newtheorem{observation}[theorem]{Observation}

\usepackage[linesnumbered,noend,ruled,vlined]{algorithm2e}
\usepackage[noend]{algpseudocode}

\SetKwInput{KwInput}{Input}
\SetKwInput{KwOutput}{Output}
\SetKw{KwTo}{to}
\SetKwFor{RepTimes}{repeat}{times}{end}
\SetKwProg{Procedure}{Procedure}{}{}

\newcommand{\M}{\mathbf M}
\newcommand{\I}{\mathcal I}
\newcommand{\B}{\mathcal B}
\newcommand{\opt}{r}
\newcommand{\OPT}{\mathsf{OPT}}

\title{Auction-Based Algorithms for Matroid Intersection: Near-Linear Query Complexity and Constant-Pass Semi-Streaming\thanks{This work is partially supported by JSPS KAKENHI Grant Numbers 
JP24K02901,
JP26K21777,
Japan, and JST ERATO Grant Number JPMJER2301, Japan.}}
\author{Chien-Chung Huang\thanks{CNRS, DIENS, PSL, France, \texttt{Chien-Chung.Huang@ens.fr}}
\and 
Yusuke Kobayashi\thanks{Kyoto University, Japan, \texttt{yusuke@kurims.kyoto-u.ac.jp}} 
}

\date{}

\begin{document}

\maketitle
\begin{abstract}
In this paper, we develop a new auction-based framework for matroid intersection and use it to obtain improved approximation algorithms in several computational settings. Our framework is inspired by Fleiner's generalized stable matching algorithm and extends the semi-streaming auction algorithm for bipartite matching due to Assadi, Liu, and Tarjan.

Using this framework, for any $\varepsilon > 0$, we present a simple $(1-\varepsilon)$-approximation algorithm in the rank-oracle model whose query complexity matches that of the current fastest algorithm. 
Furthermore, by extending this result, we obtain the first $(1-\varepsilon)$-approximation algorithm for the weighted problem 
that requires only a near-linear number of rank-oracle queries, achieving the best known rank-oracle query complexity for the problem.

We also obtain a $(1-\varepsilon)$-approximation semi-streaming algorithm for matroid intersection in the multi-pass streaming model, where the elements of the ground set arrive sequentially. It is the first algorithm achieving this approximation ratio using a constant number of passes and nearly linear space in the ranks of the matroids. When viewed in the standard offline setting, the same algorithm yields the first deterministic $(1-\varepsilon)$-approximation algorithm for matroid intersection that requires only a near-linear number of independence-oracle queries.
\end{abstract}

\section{Introduction}

In the matroid intersection problem, two matroids $\M_1 = (E, \I_1)$ 
and $\M_2 = (E, \I_2)$ are given. The goal is to compute a common independent set $S \in \I_1 \cap \I_2 $ of maximum size. First solved by Edmonds~\cite{edmonds1970matroid,edmonds1979matroid} in the 1970s, 
owing to its expressive power in capturing a wide range of fundamental combinatorial optimization problems, the matroid intersection problem has become a cornerstone of algorithm design and combinatorial optimization.
Even when restricting attention to fast algorithms for matroid intersection, there is a large body of work on the subject~\cite{blikstad2021breaking_STOC,chakrabarty2019faster,chekuri2016fast,cunningham1986improved,lawler1975matroid,nguyen2019note,BT25,FZ95,SI95}.

For the input matroids $\M_1$ and $\M_2$, we are given access to either \emph{independence oracles} or \emph{rank oracles}: given a subset $S \subseteq E$, the former answers whether $S$ is independent, while the latter returns the rank of $S$. We note that the availability of such oracles is a standard assumption in the literature, 
since explicit representations of the families $\I_1$ and $\I_2$ of independent sets are prohibitively expensive in terms of space. 
The number of queries used by an algorithm is one of the most important measures of its efficiency.

\subsection{Previous Work}
\paragraph{Fast Approximation Algorithms.}

For the matroid intersection problem, not only exact algorithms but also faster approximation algorithms have been studied extensively. 
In particular, considerable attention has been devoted to obtaining a $(1-\varepsilon)$-approximate solution, 
for an arbitrary $\varepsilon > 0$, using as few oracle queries as possible. 

Chekuri and Quanrud~\cite{chekuri2016fast} observed that a $(1-\varepsilon)$-approximate solution can be obtained 
using $O(n\opt /\varepsilon)$ independence queries by terminating Cunningham's exact algorithm~\cite{cunningham1986improved} early. 
Here, $n$ denotes the size of the ground set and $\opt$ denotes the size of a maximum common independent set.
Chakrabarty, Lee, Sidford, Singla, and Wong~\cite{chakrabarty2019faster} presented a binary-search technique and obtained a $(1-\varepsilon)$-approximation algorithm that uses $\tilde{O}(n^{3/2}/\varepsilon^{3/2})$ independence queries, 
where $\tilde{O}(\cdot)$ omits factors polynomial in $\log n$. 
They also obtained an algorithm that uses $O(n \log n / \varepsilon)$ rank queries.
Blikstad~\cite{blikstad2021breaking} improved the algorithm in~\cite{chakrabarty2019faster} so that it requires $\tilde{O}(n \sqrt{\opt} / \varepsilon)$ independence queries. 
Quanrud~\cite{Quanrud24} presented a new sparsification technique and obtained a randomized $(1-\varepsilon)$-approximation algorithm that uses $\tilde{O}(n/\varepsilon + \opt^{3/2}/\varepsilon^3)$ independence queries.
The current best result is due to Blikstad and Tu~\cite{BT25}.
They first presented a deterministic $(1-\varepsilon)$-approximation algorithm that uses $O(n^2 / (\opt  \varepsilon^2))$ independence oracle queries, which is not near-linear in $n$. By combining their approach with a randomized subroutine from~\cite{Quanrud24}, they further improved the query complexity to 
$\tilde{O}({n} /{\varepsilon} + {\opt} / {\varepsilon^5})$. 
As a result, they obtained a randomized algorithm that uses a near-linear number of independence oracle queries.
See Table~\ref{tab:summary_approximation} for a summary of $(1-\varepsilon)$-approximation algorithms. 

\begin{table}[tb]
\centering
\caption{$(1-\varepsilon)$-approximation algorithms for matroid intersection.}
\label{tab:summary_approximation}
\renewcommand{\arraystretch}{1.2}
\begin{tabular}{lccc}
\toprule
Algorithm
& \#Queries
& Ind./Rank
& Det./Rand. \\

 \midrule

Cunningham~\cite{cunningham1986improved} (observed in~\cite{chekuri2016fast})
& $O(n \opt / \varepsilon)$
& Ind.
& Det. \\

Chakrabarty--Lee--Sidford--Singla--Wong~\cite{chakrabarty2019faster}
& $O({n \log n}/{\varepsilon})$
& Rank
& Det. \\

& $\tilde{O}({n^{3/2}}/{\varepsilon^{3/2}})$
& Ind.
& Det. \\

Blikstad~\cite{blikstad2021breaking}
& $\tilde{O}(n \sqrt{\opt} / \varepsilon)$
& Ind.
& Det. \\

Quanrud~\cite{Quanrud24}
& $\tilde{O}(n / \varepsilon + \opt^{3/2} / \varepsilon^3)$
& Ind.
& Rand. \\

Blikstad--Tu~\cite{BT25}
& $O(n^2 / (\opt  \varepsilon^2))$
& Ind.
& Det. \\

& $\tilde{O}({n} /{\varepsilon} + {\opt} / {\varepsilon^5})$
& Ind.
& Rand. \\

\textbf{Theorem~\ref{thm:mainrank}}
& \textbf{$O({n \log n} / {\varepsilon})$}
& \textbf{Rank}
& \textbf{Det.} \\

\textbf{Theorem~\ref{thm:mainstreaming}}
& \textbf{$O({n \log \opt} / {\varepsilon^{2}})$}
& \textbf{Ind.}
& \textbf{Det.} \\
\bottomrule
\end{tabular}
\end{table}

Fast approximation algorithms have also been studied for the \emph{weighted} matroid intersection problem, in which each element $e \in E$ is assigned a positive weight $w(e)$, and the goal is to find a common independent set $S \in \I_1 \cap \I_2$ that maximizes $w(S) := \sum_{e \in S} w(e)$.
For this problem, Chekuri and Quanrud~\cite{chekuri2016fast} presented a $(1-\varepsilon)$-approximation algorithm 
using $\tilde{O}(n\opt/\varepsilon^2)$ independence queries, 
while Huang, Kakimura, and Kamiyama~\cite{huang2016exact} presented one using $\tilde{O}(n\opt^{3/2}/\varepsilon)$ independence queries.
Using the sparsification technique, Quanrud~\cite{Quanrud24} presented a randomized $(1-\varepsilon)$-approximation algorithm using
$O(n \log n / \varepsilon + \opt^{3/2} / \varepsilon^4)$ independence queries. 
Very recently, Dudeja and Grilnberger~\cite{DG26} introduced a reduction technique that transforms approximation algorithms for the unweighted problem into those for the weighted problem. 
Their reduction incurs an additional multiplicative factor that is exponential in $1/\varepsilon$ and depends on the range of the weights in the query complexity.

\paragraph{Streaming Model.}

In the multi-pass streaming model, 
the data is presented sequentially and the goal is to solve the problem at hand (approximately), with the ideal of using small storage space and a small number of passes over the data. Such a model was proposed as a response to the challenges of big data; see~\cite{Muthukrishnan2005}.

The streaming model for the bipartite matching problem, a special case of matroid intersection, has attracted considerable attention in the theoretical computer science community (see e.g.~\cite{assadi2024simple}). In this model, the edges of a bipartite graph arrive one by one as a data stream, and the goal is to compute an (almost) maximum matching. Since the number of edges may be enormous, algorithms are required to use only nearly linear space in the number of vertices. Such algorithms are often called \emph{semi-streaming algorithms}~\cite{FEIGENBAUM2005207}. 
A central problem in this area is to compute a $(1-\varepsilon)$-approximate maximum matching using $(1/\varepsilon)^{O(1)}$ passes and nearly linear space in the number of vertices, for any $\varepsilon>0$~\cite{AG11,assadi2021auction,EggertKMS12}. 

In the semi-streaming setting of the matroid intersection problem, we assume that the ground set $E$ is presented sequentially. 
For the matroids, we are given access to {independence oracles}. 
In this model, we need an algorithm that uses only nearly linear space in the matroid ranks.

See Table~\ref{tab:summary_streaming} for a summary of known semi-streaming algorithms for matroid intersection with a small number of passes.
Prior to this paper, only Quanrud~\cite{Quanrud24} achieved a $(1-\varepsilon)$-approximation, using the multiplicative-weight-update method. We remark that although the algorithm in~\cite{Quanrud24} is not explicitly presented as a semi-streaming algorithm, it can be adapted to the semi-streaming setting straightforwardly: one only needs to store the dual solutions generated throughout the execution of the algorithm. However, Quanrud's algorithm requires $O({\log n}/{\varepsilon})$ passes, and thus the number of passes depends logarithmically on the input size $n$.
More recently, Terao~\cite{terao:LIPIcs.WADS.2025.50} presented a constant-pass semi-streaming algorithm that achieves only a $({2}/{3}-\varepsilon)$-approximation. 

\begin{table}[tb]
\centering
\caption{Semi-streaming algorithms for matroid intersection.}
\label{tab:summary_streaming}
\renewcommand{\arraystretch}{1.2}
\begin{tabular}{lcccc}
\toprule
Algorithm
& Passes
& Space
& Approximation
& Det./Rand. \\
\midrule
Greedy
& $1$
& $O(\opt)$
& ${1}/{2}$
& Det. \\

Terao~\cite{terao:LIPIcs.WADS.2025.50}\footnotemark
& $O({1}/{\varepsilon})$
& $O({\opt}/{\varepsilon})$
& ${2}/{3}-\varepsilon$
& Det. \\

Quanrud~\cite{Quanrud24}
& $O({\log n}/{\varepsilon})$
& $O({\opt \log^2 n}/{\varepsilon})$
& $1-\varepsilon$
& Rand. \\

\textbf{Theorem~\ref{thm:mainstreaming}}
& \textbf{$O({1}/{\varepsilon^2})$}
& \textbf{$O(\opt)$}
& \textbf{$1-\varepsilon$}
& \textbf{Det.} \\
\bottomrule
\end{tabular}
\end{table}

The weighted matroid intersection problem in the semi-streaming setting has also been studied in the literature. 
Garg, Jordan, and Svensson~\cite{Garg23} presented a one-pass semi-streaming $(1/2-\varepsilon)$-approximation algorithm that requires $O(r_{\min}^2/\varepsilon^3)$ space, where $r_{\min}$ denotes the minimum of the ranks of $\M_1$ and $\M_2$.
Quanrud~\cite{Quanrud24} presented a randomized $(1-\varepsilon)$-approximation algorithm 
that requires $O(\log n/\varepsilon)$ passes and $O(\opt \log^2 n/\varepsilon)$ space.
Dudeja and Grilnberger~\cite{DG26} applied their reduction technique to Terao's algorithm~\cite{terao:LIPIcs.WADS.2025.50} and obtained a semi-streaming $(2/3-\varepsilon)$-approximation algorithm.

A further generalization of (weighted) matroid intersection is the problem of maximizing a monotone submodular function subject to a matroid constraint~\cite{CalinescuCPV11,FilmusW14}. 
This problem has also been studied in the semi-streaming setting; see, for example,~\cite{ChakrabartiK15,ChekuriGQ15,FeldmanLNSZ26,HuangTW20}.

\footnotetext{
Although~\cite{terao:LIPIcs.WADS.2025.50} states a space requirement of $O({r_{\rm sum}}/{\varepsilon})$, where $r_{\rm sum}$ is the sum of the ranks of $\M_1$ and $\M_2$, this can be easily improved to $O({\opt}/{\varepsilon})$ by truncating each matroid to rank $2\opt$ (see Section~\ref{sec:definition}).
}

\subsection{Our Contributions}

In this paper, we present auction-based algorithms, in which $E$ is regarded as a set of items and the price of each $e \in E$ increases monotonically.  
Our first contribution is a simple algorithm that yields the following result. 

\begin{restatable}{theorem}{mainrank}
\label{thm:mainrank}
Let $\varepsilon > 0$.
For the matroid intersection problem, 
there exists a deterministic $(1-\varepsilon)$-approximation algorithm
using $O({n \log n} / {\varepsilon})$ rank queries.
\end{restatable}

This is currently the fastest $(1-\varepsilon)$-approximation algorithm for matroid intersection, measured by the number of rank oracle queries, if one is allowed to use rank oracles.  
We note that an algorithm with exactly the same query complexity is also given in~\cite{chakrabarty2019faster}; see Table~\ref{tab:summary_approximation}.
Our algorithm is based on an auction mechanism and is quite simple, whereas the algorithm in~\cite{chakrabarty2019faster} relies on the exchange graph.

Our second contribution is to extend this result 
to the weighted matroid intersection problem while maintaining the same query complexity. 
This is the first $(1-\varepsilon)$-approximation algorithm for the weighted matroid intersection problem
that requires only a near-linear number of rank queries, achieving the best known rank-oracle query complexity for the problem.

\begin{restatable}{theorem}{mainrankweighted}
\label{thm:mainrankweighted}
Let $\varepsilon > 0$.
For the weighted matroid intersection problem, 
there exists a deterministic $(1-\varepsilon)$-approximation algorithm 
using $O({n \log n} / {\varepsilon})$ rank queries.
\end{restatable}

Note that, by applying the reduction technique of Dudeja and Grilnberger~\cite{DG26} 
to the near-linear rank-query algorithm in~\cite{chakrabarty2019faster}, 
one obtains a $(1-\varepsilon)$-approximation algorithm for weighted matroid intersection. 
However, as described above, the reduction incurs an additional multiplicative factor in the query complexity 
that depends on the range of the weights. 
As a result, the resulting algorithm no longer uses a near-linear number of rank queries.

Based on this algorithm, we present another auction-based algorithm that works in the semi-streaming setting. 
Given the connection to bipartite matching, a natural goal in semi-streaming matroid intersection is to compute a $(1-\varepsilon)$-approximate common independent set using $(1/\varepsilon)^{O(1)}$ passes and nearly linear space in the matroid ranks, for any $\varepsilon>0$. 
Our third contribution is to present the first algorithm that achieves this goal.

\begin{theorem}\label{thm:mainstreaming}
Let $\varepsilon > 0$.
For the matroid intersection problem in the semi-streaming setting, 
there exists a deterministic $(1-\varepsilon)$-approximation algorithm 
using 
$O({1}/{\varepsilon^{2}})$ passes,  
$O(\opt)$ space\footnote{In the present work, we measure the space complexity by the number of machine words and we assume that each machine word has $O(\log n)$ bits. The space complexity increases by an additional factor of $O(\log n)$ when measured in bits.},  
and $O({n \log \opt}/{\varepsilon^{2}})$ independence queries.
\end{theorem}

It is worth noting that our algorithm is deterministic, and it has a smaller space requirement than 
Quanrud's algorithm~\cite{Quanrud24}. 
Furthermore, even in the offline setting, our algorithm is competitive with the current fastest $(1-\varepsilon)$-approximation algorithm in terms of running time, measured by the number of independence oracle queries.
As shown in Table~\ref{tab:summary_approximation}, our algorithm is the first \emph{deterministic} $(1-\varepsilon)$-approximation algorithm for matroid intersection that achieves a near-linear number of independence oracle queries.

By combining Theorem~\ref{thm:mainstreaming} with the reduction technique that transforms unweighted algorithms into weighted ones (see Lemma 4.5 in the arXiv version of~\cite{DG26}), we obtain the following result for the weighted problem.

\begin{corollary}\label{cor:weighted}
Let $\varepsilon > 0$.
For the weighted matroid intersection problem in the semi-streaming setting, 
there exists a deterministic $(1-\varepsilon)$-approximation algorithm  
using 
$O({1}/{\varepsilon^{2}})$ passes and 
$O\big( \opt \cdot \lceil {1}/{\varepsilon} \rceil ^{{1}/{\varepsilon}} \cdot \log W \big)$ space, 
where $W :=\frac{\max_{e\in E} w(e)}{\min_{e\in E} w(e)}$ denotes the aspect ratio of the weights.
\end{corollary}

This corollary yields a $(1-\varepsilon)$-approximation algorithm that requires only a constant number of passes, 
at the cost of an additional $\log W$ factor in the space complexity.
The resulting trade-off is incomparable to that of the randomized algorithm of Quanrud~\cite{Quanrud24}.

\paragraph{Back to Bipartite Graph Matching.}

The algorithm in Theorem~\ref{thm:mainstreaming} can be applied directly to bipartite matching without any modification.
This yields another $(1-\varepsilon)$-approximation semi-streaming algorithm for bipartite matching.
Below, we compare its performance with that of other known $(1-\varepsilon)$-approximation semi-streaming algorithms.
Note that, in the bipartite matching setting, 
$E$ corresponds to the edge set of the input graph and 
the maximum common independent set size $\opt$ corresponds to the size of a maximum matching.
Let $r'$ denote the number of vertices in the input graph. 
It is obvious that $r \le r'$.

As mentioned earlier, there is a vast literature on graph matching in the semi-streaming setting, and we do not attempt to provide an exhaustive survey here.
We refer interested readers to Assadi's paper~\cite{assadi2024simple}.
A central theme in this line of research is to achieve both a small number of passes and low space complexity.
Our algorithm for bipartite matching obtained from Theorem~\ref{thm:mainstreaming} performs favorably with respect to both measures. 

In terms of the number of passes, the current best $(1-\varepsilon)$-approximation algorithms require either $O({1}/{\varepsilon^2})$ passes~\cite{AG11,assadi2021auction} or $O({\log r'}/{\varepsilon})$ passes~\cite{AG18,assadi2024simple,AJJST22}.
Our algorithm and Quanrud's algorithm~\cite{Quanrud24} for matroid intersection may be viewed as the corresponding counterparts of these two approaches.

In terms of space complexity, the best known algorithms require $O(r')$ space~\cite{assadi2024simple,assadi2021auction,AJJST22,EggertKMS12}. Our algorithm further improves this to $O(\opt)$ space. This improvement is made possible by the standard trick of truncating the two given partition matroids to rank $O(\opt)$ (see Section~\ref{sec:definition}), thereby transforming them into laminar matroids. We believe that this illustrates the advantage of studying the problem from a more abstract perspective.

\subsection{Our Technique}

Most algorithms for matroid intersection are based on finding augmenting paths in a certain exchange graph; see, e.g.,~\cite{blikstad2021breaking_STOC,chakrabarty2019faster,chekuri2016fast,cunningham1986improved,lawler1975matroid,nguyen2019note}. 
Rare exceptions include algorithms based on auction mechanisms~\cite{BT25,FZ95,SI95}. 
Our present algorithms also follow the auction paradigm.

More precisely, our algorithms are based on the auction algorithm of Assadi, Liu, and Tarjan~\cite{assadi2021auction} for bipartite matching. 
In particular, our algorithm for Theorem~\ref{thm:mainstreaming} generalizes their algorithm; 
thus, the similarity in the number of passes is not coincidental.
Our main insight is that there is a close connection between the auction algorithm of~\cite{assadi2021auction} and the stable matching algorithm of Gale and Shapley~\cite{GaleShapley1962}. 
This observation suggests that, in order to develop an auction algorithm for matroid intersection, one should consider stable matching with matroid constraints, which fits into a framework introduced by Fleiner~\cite{Fleiner03}. 
From this perspective, our algorithm can be viewed as an adaptation of the algorithm of Assadi, Liu, and Tarjan~\cite{assadi2021auction} to Fleiner's framework~\cite{Fleiner03}.
The connection between our matroid intersection algorithm and generalized stable matching is discussed in Section~\ref{sec:relationtoGSM}.

We now describe an informal outline of the algorithms for Theorems~\ref{thm:mainrank}, \ref{thm:mainrankweighted}, and~\ref{thm:mainstreaming}.
In the proposed algorithms,
$E$ is regarded as a set of items, and the price  of each item increases monotonically.

There are two players, Alice and Bob:
Alice aims to buy as many elements as possible from $\I_1$ at low prices,
while Bob aims to sell as many elements as possible from $\I_2$ at high prices.
Throughout the algorithm, we maintain the prices and a set $X \subseteq E$ representing the items
provisionally agreed upon for trade between Alice and Bob.
 
We begin with $X=\emptyset$ and repeat the following procedure. In each iteration, if Alice wishes to buy additional items beyond those in $X$, she proposes replacing $X$ with a superset $Y$. Bob then removes some low-priced elements from $Y$ so that the resulting set belongs to $\I_2$ and increases the prices of the removed items. If Alice does not wish to buy any element in $E \setminus X$, the algorithm terminates and returns $X$.

By implementing this mechanism in a straightforward manner,
we obtain a simple algorithm that uses a near-linear number of rank queries.
While it is not immediately obvious that the resulting solution is a $(1-\varepsilon)$-approximation,
this can be proved using techniques in matroid optimization,
such as the weight-splitting technique.
These ideas lead to Theorems~\ref{thm:mainrank} and~\ref{thm:mainrankweighted}. 

To obtain a semi-streaming algorithm and prove Theorem~\ref{thm:mainstreaming}, more careful arguments are required.
For example, a naive implementation would require $\Theta(n)$ space to maintain the prices.
To address this issue, we develop an implicit representation of prices, which constitutes one of the key ideas of our algorithm.
Roughly speaking, we maintain prices only for the elements in $X$, thereby reducing the space complexity.
Then, for each $e \in E \setminus X$, its price is defined as the price at which Bob is willing to sell $e$, 
possibly after removing some element from $X$.

\subsection{Paper Organization}

The remainder of this paper is organized as follows.
In Section~\ref{sec:preliminary}, we introduce notation and review basic properties of matroids.
To provide intuition for the auction mechanism in the matroid intersection setting,
Section~\ref{sec:rankquery} presents a simple auction algorithm using a nearly linear number of rank queries and proves Theorem~\ref{thm:mainrank}.
Section~\ref{sec:rankweighted} extends this result to the weighted setting and proves Theorem~\ref{thm:mainrankweighted}.
Building on the algorithm of Theorem~\ref{thm:mainrank}, Section~\ref{sec:semistreaming} presents a semi-streaming algorithm and proves Theorem~\ref{thm:mainstreaming}.
Finally, Section~\ref{sec:conclusion} concludes the paper.  

\section{Preliminary}
\label{sec:preliminary}

\subsection{Notation}
\label{sec:definition}
Let $\mathbb{R}$ and $\mathbb{R}_+$ denote the sets of real numbers and non-negative real numbers, respectively. 
Let $\mathbb{Z}_+$ denote the set of non-negative integers. 
For a set $X$ and an element $e$, $X \cup \{e\}$ and $X \setminus \{e\}$ are simply denoted by $X+e$ and $X-e$, respectively. 

A {\em matroid} is a pair $\M = (E, \I)$ consisting of a finite ground set $E$ and 
a non-empty family $\I \subseteq 2^E$ of subsets of $E$ such that 
\begin{itemize}
    \item if $X \subseteq Y$ and $Y \in \I$, then $X \in \I$; and 
    \item if $X\in \I$, $Y \in \I$, and $|Y| > |X|$, then there exists $e \in Y \setminus X$ such that $X + e \in \I$. 
\end{itemize}
The sets in $\I$ are referred to as {\em independent sets}. 
An element $e \in E$ is called a {\em loop} if $\{e\} \not\in \I$. 
An inclusion-wise maximal independent set is called a {\em base} of $\M$. 
It is easy to see that every base of $\M$ has the same size. 
It is well known that the family of bases satisfies the {\em simultaneous exchange property}: 
for any two bases $B_1$ and $B_2$ of $\M$ and any element $e \in B_1 \setminus B_2$, 
there exists an element $f \in B_2 \setminus B_1$ 
such that both $B_1 - e + f$ and $B_2 + e - f$ are bases of $\M$; see e.g.~\cite[Theorem~39.12]{Schrijver2003}. 
An inclusion-wise minimal non-independent set is called a {\em circuit} of $\M$. 
For an independent set $X \in \I$ and an element $e \in E \setminus X$ such that $X + e \not\in \I$, 
$X + e$ contains a unique circuit, which we denote $C_{\M}(X \mid e)$, and it is called the 
the \emph{fundamental circuit} of $X$ and $e$.
For a matroid $\M=(E, \I)$, we define the {\em rank function} ${\rm rank}_{\M} \colon 2^E \to \mathbb{Z}_+$ 
and the {\em span function} ${\rm span}_{\M} \colon 2^E \to 2^E$ by 
\begin{align*}
{\rm rank}_{\M} (X) &= \max \{ |Y| \mid Y \subseteq X,\ Y \in \I\}, \\
{\rm span}_{\M} (X) &= \{ e \in E \mid {\rm rank}_{\M} (X +e) = {\rm rank}_{\M} (X) \}
\end{align*}
for any $X \subseteq E$. 
See~\cite{oxley2011matroid,Schrijver2003} for more on matroids.

For a set $X \subseteq E$ and a vector $p \in \mathbb{R}^E$, we denote $\sum_{e \in X} p(e)$ by $p(X)$. 
For a matroid $\M = (E, \I)$ and a vector $p \in \mathbb{R}^E$, 
a {\em $p$-minimum base of $\M$} is a base $B$ of $\M$ that minimizes $p(B)$. 
We define a {\em $p$-minimum independent set} and a {\em $p$-maximum base/independent set} in a similar way. 
It is well known that such sets can be computed by a greedy algorithm.

To discuss the algorithmic problem for matroids, we suppose that each matroid is given as {\em oracles}. 
An {\em independence oracle} for a matroid $\M$ is an oracle that, given a subset $X \subseteq E$, determines whether $X \in \I$.
A {\em rank oracle} for a matroid $\M$ is an oracle that, given a subset $X \subseteq E$, returns ${\rm rank}_{\M} (X)$. 
Note that we can determine whether $X \in \I$ by using the rank oracle, which implies that the rank oracle is at least as powerful as the independence oracle.

The main focus of this paper is a fast $(1-\varepsilon)$-approximation algorithm for the matroid intersection problem, where $\varepsilon > 0$ is a given constant. 
In the problem, we are given a finite set $E$ of size $n$ and two matroids $\M_1 = (E, \I_1)$ and $\M_2 = (E, \I_2)$ over $E$.  
Let ${\OPT}$ be a common independent set of maximum size, and let $\opt = |{\OPT}|$. 
The goal is to find a common independent set of size at least $(1-\varepsilon) \opt$. 
For $i \in \{1, 2\}$, the rank function ${\rm rank}_{\M_i}(\cdot)$ and the fundamental circuit $C_{\M_i}(\cdot)$ 
are denoted by $r_{i}(\cdot)$ and $C_{i}(\cdot)$, respectively, and 
let $\B_i$ denote the base family of $\M_i$. 
It is well known (and easy to show) that the greedy algorithm yields a common independent set of size at least $\opt / 2$.

Throughout the paper, without loss of generality, we assume the following: 
\begin{itemize}
\item ${1}/{\varepsilon}$ is an integer, because
we can replace $\varepsilon$ with a smaller constant $\lceil \varepsilon^{-1} \rceil^{-1}$. 
\item ${1}/{\varepsilon} \le n+1$, since otherwise a $(1-\varepsilon)$-approximation algorithm has to find an exact solution, 
and so we can replace $\varepsilon$ with $\frac{1}{n+1}$. 
\item $\M_1$ and $\M_2$ have no loops, 
since otherwise we can remove them from $E$. 
\item 
$r_i (E) \le 2 \opt$ for $i \in \{1, 2\}$, because we can truncate each matroid to rank
$2 |X^*|$, where $X^*$ is a 2-approximate solution obtained by the greedy algorithm. 
\end{itemize}
Here, for $k \in \mathbb{Z}_+$,  
the {\em truncation of $\M = (E, \I)$ to rank $k$} is a pair $\M' = (E, \I')$ defined by 
$\I' := \{X \in \I \mid |X| \le k\}$, which is known to be a matroid.

\subsection{Basic Properties}

In this subsection, we present basic properties concerning matroids. 
First, the following observation follows from the uniqueness of the fundamental circuit. 

\begin{observation}\label{obs:fundamentalcircuit}
Let $\M = (E, \I)$ be a matroid, 
let $X \in \I$, and let $e \in E \setminus X$ be such that $X + e \notin \I$. 
Then, for $f \in X + e $, we have $X + e - f \in \I$ if and only if $f \in C_{\M}(X \mid e)$.
\end{observation}

The following lemma is known as the gross substitute property.
We provide a proof for completeness.

\begin{lemma}[\mbox{see e.g.~\cite[Proposition 6.32]{murota2003discrete}}] \label{lem:gsproperty}
    Let $\M=(E, \I)$ be a matroid and let $S \subseteq E$. 
    Let $p, p' \in \mathbb{R}^E$ be vectors such that $p(e) \le p'(e)$ for all $e \in E$, and $p(e) = p'(e)$ for all $e \in S$.
    If $S$ is contained in a $p$-minimum base of $\M$, then it is also contained in a $p'$-minimum base of $\M$.
\end{lemma}

\begin{proof}
Let $B$ be a $p$-minimum base of $\M$ that contains $S$.
Let $B'$ be a $p'$-minimum base of $\M$ that minimizes $|S \setminus B'|$.
If $|S \setminus B'| = 0$, then the claim holds.
Assume, to the contrary, that $|S \setminus B'| \ne 0$, that is, there exists an element $e \in S \setminus B'$.
By applying the simultaneous exchange property to $B$ and $B'$ with respect to $e$, there exists an element $e' \in B' \setminus B$ 
such that $B - e + e'$ and $B' + e - e'$ are bases of $\M$.
Since $B$ is a $p$-minimum base, we have $p(B) \le p(B - e + e')$, and hence $p(e) \le p(e')$.
Therefore, $p'(e) = p(e) \le p(e') \le p'(e')$, because $e \in S$.
Hence, $B'' := B' + e - e'$ satisfies that $p'(B'') = p'(B') + p'(e) - p'(e') \le p'(B')$.
Thus, $B''$ is a $p'$-minimum base of $\M$ such that $|S \setminus B''| = |S \setminus B'| - 1$,
which contradicts the choice of $B'$.
\end{proof}

The next lemma below states that an independent set of a matroid is $w$-maximum if it is $w$-maximum among its local neighbors.

\begin{lemma}[\mbox{see~\cite[Corollary 8.6]{cook2011combinatorial}}] \label{lem:charaminind}
    Let $\M=(E, \I)$ be a matroid and let $w \in \mathbb{R}^E$. 
    For $S \in \I$, $S$ is a $w$-maximum independent set of $\M$ if and only if 
    (i)  $w(e) \le 0$ for any $e \in E \setminus S$ with $S+e \in \I$,
    (ii) $w(e) \ge 0$ for any $e \in S$, and 
    (iii) $w(e) \le w(f)$ for any $e \in E \setminus S$ and $f \in S$ with $S+e-f \in \I$. 
\end{lemma}

Nguy$\tilde{{\hat{\text{e}}}}$n~\cite{nguyen2019note} and Chakrabarty--Lee--Sidford--Singla--Wong \cite{chakrabarty2019faster} 
independently developed the so-called {\em binary search technique}, 
which is useful for improving the query complexity of matroid intersection algorithms.

\begin{lemma}[\mbox{Binary search technique~\cite{nguyen2019note,chakrabarty2019faster}}]
\label{lem:binary}
Let $\M=(E, \I)$ be a matroid and let $w \in \mathbb{R}^E$ be a weight function. Let $S \in \I$, $T \subseteq S$, and $v \in E \setminus S$. Then the following hold.
\begin{enumerate}
\item[(1)] Using $O(\log |T|)$ independence queries, we can find an element $u \in T$ minimizing $w(u)$ subject to $S + v - u \in \I$ (if such an element exists).
\item[(2)] Using $O(\log |E \setminus S|)$ rank queries, we can find an element $u \in E \setminus S$ minimizing $w(u)$ subject to $S + u \in \I$ (if such an element exists).
\end{enumerate}
\end{lemma}

\section{Simple Algorithm with Nearly Linear Rank Queries}
\label{sec:rankquery}

In this section, we present a simple auction algorithm for matroid intersection and prove Theorem~\ref{thm:mainrank}. 
Our algorithm is presented in Algorithm~\ref{alg:ranklinear}. 
For clarity, we provide an informal explanation of the algorithm below.  

The algorithm is of auction type, in which $E$ is the set of items and the price $p(e)$ for each $e \in E$ increases monotonically.  
There are two players, Alice and Bob: Alice aims to buy as many elements as possible from $\I_1$, each with price at most $1$, 
while Bob aims to sell as many elements as possible from $\I_2$ at high prices.  
The set $X$ represents the items 
provisionally
agreed upon for trade between Alice and Bob. 

In each iteration, if Alice wishes to buy an element $e \in E \setminus X$ in addition to $X$ at price $p(e)$, she proposes to Bob to replace $X$ with $X + e$.  
If $X + e \in \I_2$, then Bob accepts the proposal. Otherwise, Bob removes the cheapest element 
$f \in X + e$ such that $X + e - f \in \I_2$ (i.e., $f \in C_{2}(X \mid e)$) and increases the price of $f$ by $\varepsilon$.  
The algorithm repeats this process until there is no element in $E \setminus X$ that Alice wishes to buy, 
and then returns $X$.

\begin{algorithm}
    \caption{Matroid Intersection Algorithm with Nearly Linear Rank Queries}\label{alg:ranklinear}
    \KwInput{Two matroids $\M_1 = (E, \I_1)$ and $\M_2 = (E, \I_2)$}
    \KwOutput{A common independent set $X \in \I_1 \cap \I_2$} 
    $X \gets \emptyset$ ; \\
	\ForEach{$e \in E$}{       
    $p(e) \gets 0$ ; \\
    }
    \While {true}{
    Find $e \in E \setminus X$ minimizing $p(e)$ subject to $X + e \in \I_1$ ; \\
    \tcp{Let $e = \bot$ if no such element exists}
    \If{$e \neq \bot$ and $p(e) \le 1$}{
        \If{$X + e \in \I_2$}{
            $X \gets X + e$ ; 
        }
        \Else{
        Find $f \in C_{2}(X \mid e)$ minimizing $p(f)$   \tcp*{possibly, $f=e$}
        $X \gets X + e - f$ ; \\                          
        $p(f) \gets p(f) + \varepsilon$ ; 
        }
    }
    \Else{
        \Return $X$ ;         
    }
    }
\end{algorithm}

\subsection{Query Complexity}
\label{sec:rankcomplexity}

In this subsection, we show that 
Algorithm~\ref{alg:ranklinear} returns a common independent set $X \in \I_1 \cap \I_2$ 
by using $O({n \log n} / {\varepsilon})$ rank queries. 
To this end, we begin with the following simple observation.

\begin{observation}\label{obs:easyobsrank}
    Throughout the execution of Algorithm~\ref{alg:ranklinear}, we have the following properties: 
    \begin{enumerate}
    \item 
    $X \in \I_1 \cap \I_2$.
    \item 
    $|X|$ is monotonically non-decreasing. 
    \item 
    $0\le p(e) \le 1$ for each $e \in X$. 
    \item 
    $0\le p(e) \le 1+\varepsilon$ for each $e \in E$. 
    \item 
    For each $e \in E$, $p(e)$ is monotonically non-decreasing.
\end{enumerate}
\end{observation}

\begin{proof}
    Properties 2 and 5 are obvious. 
    Since $X \in \I_1$ follows from the update rule and $X \in \I_2$ is preserved in Line 11 by Observation~\ref{obs:fundamentalcircuit}, 
    Property 1 holds. 
    Property 3 is preserved, because $e \in E \setminus X$ is added to $X$ only when $p(e) \le 1$ holds in Line 6. 
    For $f \in E$, observe that $p(f)$ increases by $\varepsilon$ only when $f$ is removed from $X+e$ in Lines 11 and 12. 
    Since the price of each element in $X+e$ is at most $1$ by Property 3, the updated price of $f$ is at most $1+\varepsilon$, 
    which shows Property 4. 
\end{proof}

Note that, by Property 1 in this observation, $C_{2}(X \mid e)$ is well-defined in Line 10. 
We next evaluate the number of rank queries used in the algorithm.

\begin{lemma}\label{lem:numberofrankqueries}
    Algorithm~\ref{alg:ranklinear} uses $O({n \log n}/{\varepsilon})$ rank queries.
\end{lemma}

\begin{proof}
By Observation~\ref{obs:easyobsrank}, we see the following: 
\begin{itemize}
    \item 
    Since $X \in \I_1 \cap \I_2$ implies $0 \le |X| \le \opt$,  
    Line 8 is executed at most $\opt$ times.
    \item 
    Since $0 \le p(E) \le (1+\varepsilon)n$ and 
    the value of $p(E)$ increases by $\varepsilon$ in Line 12, 
    this line is executed at most ${(1+\varepsilon)n}/{\varepsilon}$ times.
\end{itemize} 
Therefore, the while-loop is executed at most $\opt + \frac{(1+\varepsilon)n}{\varepsilon}  + 1 = O({n}/{\varepsilon})$ times. 
We also see that each line can be implemented efficiently as follows: 
\begin{itemize}
    \item By applying Lemma~\ref{lem:binary} (2) with $S = X$ and $w = p$, we can implement Line 5 using $O(\log n)$ rank queries to $\M_1$.  
    \item By applying Lemma~\ref{lem:binary} (1) with $S = T = X$, $v = e$, and $w = p$, we can implement Line 10 using $O(\log \opt)$ independence queries to $\M_2$. 
\end{itemize}

Since the rank query is as powerful as the independence query, 
the algorithm runs with $O({n \log n}/{\varepsilon})$ rank queries in total. 
\end{proof}

By Observation~\ref{obs:easyobsrank} and Lemma~\ref{lem:numberofrankqueries}, 
Algorithm~\ref{alg:ranklinear} uses $O({n \log n}/{\varepsilon})$ rank queries and returns a common independent set $X \in \I_1 \cap \I_2$.

\subsection{Key Invariants}

We next show key invariants in Algorithm~\ref{alg:ranklinear} 
that are used to guarantee the approximation ratio.

\begin{lemma}\label{lem:propertyrank1}
At the beginning of each iteration of the while-loop in Algorithm~\ref{alg:ranklinear}, 
$X$ is contained in a $p$-minimum base of $\M_1$.  
\end{lemma}

\begin{proof}
Obviously, $X=\emptyset$ satisfies the condition at the beginning of the algorithm. 
To prove the lemma by induction, 
suppose that $X$ is contained in a $p$-minimum base, say $Z$, of $\M_1$ at the beginning of some iteration of the while-loop. 
The following claim is a basic property of matroids, but we include a proof for completeness.

\begin{claim}
If $e \in E \setminus X$ is chosen in Line 5, where $e \neq \bot$, then 
$X + e$ is contained in a $p$-minimum base of $\M_1$. 
\end{claim}

\begin{proof}[Proof of the claim]
    Since the claim is obvious if $e \in Z$, it suffices to consider the case where $e \not\in Z$. 
    In this case, we have $Z+e \not\in \I_1$. 
    Since $C_1(Z \mid e)\not\in \I_1$ and  $X+e \in \I_1$, there exists an element $e' \in C_1(Z \mid e) \setminus (X+e)$. 
    Then, $Z' := Z + e - e'$ is a base of $\M_1$ by Observation~\ref{obs:fundamentalcircuit}. 
    Furthermore, $p(e') \ge p(e)$ holds by the choice of $e$, where we note that $X+e' \subseteq Z \in \I_1$. 
    This implies that $p(Z') \le p(Z)$, that is, $Z'$ is a $p$-minimum base of $\M_1$. 
    Since $X+e$ is contained in $Z'$, the claim holds.  
\end{proof}

By this claim, when $X$ is replaced by $X+e$ or $X+e-f$ in Line 8 or 11, 
the resulting set is also contained in a $p$-minimum base of $\M_1$.  
Furthermore, when $p(f)$ is replaced by $p(f)+\varepsilon$ in Line 12, 
$X$ is still contained in a $p$-minimum base of $\M_1$ by Lemma~\ref{lem:gsproperty}, because $f \not\in X$. 
Therefore, $X$ is contained in a $p$-minimum base in the next iteration of the while-loop. 
By induction, this completes the proof of Lemma~\ref{lem:propertyrank1}. 
\end{proof}

\begin{lemma}
\label{lem:propertyrank2}
At the beginning of each iteration of the while-loop in Algorithm~\ref{alg:ranklinear}, 
for any $e' \in E \setminus X$ with $p(e') \ge \varepsilon$, 
we have that 
\begin{itemize}
\item 
$X+e' \not\in \I_2$, and 
\item 
$p(e') \le p(f') + \varepsilon$ for any $f' \in C_2(X \mid e')$.
\end{itemize}
\end{lemma}

\begin{proof}
 For a positive integer $t$, suppose that $X$ is updated from $X^t$ to $X^{t+1}$ in the $t$-th iteration of the while-loop. 
 The elements $e$ and $f$ chosen in this iteration are denoted by $e^t$ and $f^t$, respectively, 
 where $f^t$ is undefined if $X^{t+1} = X^t + e^t$. 
 The price vector $p$ at the beginning of this iteration is denoted by $p^t$. 

 To prove the lemma, fix an iteration number $t^*$ and an element $e' \in E\setminus X^{t^*}$ with $p^{t^*}(e') \ge \varepsilon$.
 The objective is to show that 
 $X^{t^*} + e' \not\in \I_2$, and 
 $p^{t^*}(e') \le p^{t^*}(f') + \varepsilon$ for any $f' \in C_2(X^{t^*} \mid e')$. 
 Let $t$ be the maximum integer subject to $t < t^*$ and $f^{t} = e'$. 
 Note that such $t$ must exist, because $p^{t^*}(e') \ge \varepsilon$. 
 In what follows, for $\tau = t+1, t+2, \dots , t^*$, we prove
 \begin{enumerate}
 \item[(i)] 
 $X^\tau + e' \not\in \I_2$, and 
 \item[(ii)] 
 $p^\tau (e') \le p^\tau (f') + \varepsilon$ for any $f' \in C_2(X^\tau \mid e')$
 \end{enumerate}
  by induction.

\begin{claim}\label{clm:36}
  For $\tau = t+1$, (i) and (ii) hold. 
\end{claim}

\begin{proof}[Proof of the claim]
 By the choice of $f^t$, we have that $X^{t} + e^{t} \not\in \I_2$ and 
 $X^{t+1} = X^{t} + e^{t} - f^t$. 
 Since $f^{t} = e'$ is a minimizer of $p^{t}$ on $C_2(X^{t} \mid e^{t})$, 
 we obtain
\[
  p^t(e') \le p^t(f')  \quad \text{for any } f' \in C_2(X^t \mid e^t). 
\]
Since the price of $e'$ increases by $\varepsilon$ in Line 12, this shows that 
\[
p^{t+1}(e') = p^{t}(e') + \varepsilon \le p^{t}(f') + \varepsilon \le p^{t+1}(f') + \varepsilon \quad \text{for any } f' \in C_2(X^t \mid e^t).
\]
Furthermore, 
$C_2(X^t \mid e^t) = C_2(X^{t+1} \mid e')$ holds, because $X^t + e^t = X^{t+1} + e'$.  
Therefore, (i) and (ii) hold for $\tau = t+1$.  
\end{proof}

\begin{claim}\label{clm:37}
    Suppose that (i) and (ii) hold for an integer $\tau$ with $t+1 \le \tau < t^*$. 
    Then they also hold for $\tau +1$. 
\end{claim}

\begin{proof}[Proof of the claim]
If $X^{\tau+1} = X^{\tau} + e^{\tau}$, then we obtain
$C_2(X^{\tau +1} \mid e') = C_2(X^{\tau} \mid e')$ and $p^{\tau +1} = p^{\tau}$, and hence the claim holds. 
Therefore, it suffices to consider the case where $X^{\tau+1} = X^{\tau} + e^{\tau} - f^{\tau}$. 
Note that $e^{\tau} \neq e'$, by the maximality of $t$ and the fact that $e' \not\in X^{t^*}$. 

If $f^{\tau} \not\in C_2(X^{\tau} \mid e')$, then 
we obtain $X^{\tau +1} + e' \not\in \I_2$ and $C_2(X^{\tau +1} \mid e') = C_2(X^{\tau} \mid e')$, and hence  
\[
p^{\tau+1}(e') = p^{\tau}(e') \le p^{\tau}(f') + \varepsilon = p^{\tau+1}(f') + \varepsilon
\]
holds for any $f' \in C_2(X^{\tau+1} \mid e')$.
Thus, (i) and (ii) hold for $\tau +1$. 

Otherwise, $f^\tau \in C_2(X^{\tau} \mid e')$. 
By the choice of $f^\tau$, it is also contained in $C_2(X^{\tau} \mid e^{\tau})$. 
Note that $C_2(X^{\tau} \mid e')$ and $C_2(X^{\tau} \mid e^{\tau})$ are two distinct circuits, because 
$e' \neq e^{\tau}$. 
Since $f^\tau$ is contained in two distinct circuits, 
by a basic property for circuits (see~\cite[Theorem 39.7]{Schrijver2003}), 
there exists a circuit $C$ of $\M_2$ such that
\[
C \subseteq (C_2(X^{\tau} \mid e') \cup C_2(X^{\tau} \mid e^{\tau})) - f^{\tau}, 
\]
which is contained in $X^{\tau +1} + e'$.  
Therefore, (i) holds for $\tau+1$, and $C = C_2(X^{\tau+1} \mid e')$. 
Combining 
     \begin{align*}
       & p^\tau(e') \le p^\tau (f') + \varepsilon  \quad \text{for any } f' \in C_2(X^{\tau} \mid e'), \\
       & f^\tau \in C_2(X^{\tau} \mid e'), \\
       & p^\tau(f^\tau) = \min \{ p^\tau(f') \mid f' \in C_2(X^\tau \mid e^\tau) \}, 
     \end{align*}
we obtain $p^\tau(e') \le p^\tau (f') + \varepsilon$ for every $f' \in C_2(X^{\tau} \mid e') \cup C_2(X^{\tau} \mid e^{\tau})$.  
Since $C=C_2(X^{\tau+1} \mid e')$ is contained in $(C_2(X^{\tau} \mid e') \cup C_2(X^{\tau} \mid e^{\tau})) -f^\tau$, we have that
\[
p^{\tau+1}(e') = p^{\tau}(e') \le p^{\tau}(f') + \varepsilon = p^{\tau+1}(f') + \varepsilon
\]
for any $f' \in C_2(X^{\tau+1} \mid e')$. 
Therefore, (ii) holds for $\tau +1$. 
\end{proof}

By Claims~\ref{clm:36} and~\ref{clm:37}, 
(i) and (ii) hold for $\tau = t+1, t+2, \dots , t^*$. 
In particular, they hold for $\tau = t^*$, which completes the proof for Lemma~\ref{lem:propertyrank2}. 
\end{proof}

\subsection{Proof for Theorem~\ref{thm:mainrank}}

We are now ready to prove Theorem~\ref{thm:mainrank}, which we restate here. 

\mainrank*

\begin{proof}
We have already seen in Section~\ref{sec:rankcomplexity} that Algorithm~\ref{alg:ranklinear} uses $O({n \log n}/{\varepsilon})$ rank queries and returns a common independent set $X \in \I_1 \cap \I_2$. 
Therefore, it suffices to show that 
the output $X$ of Algorithm~\ref{alg:ranklinear} has size at least $(1-\varepsilon) \opt$. 

Let $q_1(e) = 1 - p(e)$ for each $e \in E$ and let
\[
q_2(e) = 
\begin{cases}
p(e) & \mbox{for $e \in X$,} \\
p(e) - \varepsilon & \mbox{for $e \in E \setminus X$.}
\end{cases}
\]
By Lemma~\ref{lem:propertyrank1}, 
there exists a $p$-minimum base $\hat X$ of $\M_1$ that contains $X$. 
By Observation~\ref{obs:easyobsrank}, for every $e \in X$, 
we have $p(e) \le 1$.
Furthermore, by the stopping condition of the algorithm, 
there exists no element $e \in E \setminus X$ such that $X + e \in \I_1$ and $p(e) \le 1$.  
Therefore, 
\[
X = \{e \in \hat X \mid p(e) \le 1\} = \{e \in \hat X \mid q_1(e) \ge 0\}. 
\]
Since $\hat X$ is a $q_1$-maximum base, this implies that $X$ is a $q_1$-maximum independent set of $\M_1$.
Therefore, $q_1(X) \ge q_1(\OPT)$, because $\OPT \in \I_1$. 

We next observe the following: 
\begin{itemize}
    \item 
    For $e \in X$, we obtain $q_2(e) = p(e) \ge 0$. 
    \item 
    For $e \in E \setminus X$ with $X + e \in \I_2$, 
    we obtain $p(e) < \varepsilon$ by Lemma~\ref{lem:propertyrank2}, 
    implying that $q_2(e) < 0$. 
    \item 
    For $e \in E \setminus X$ and $f \in X$ such that $X + e \in \I_2$ (which implies that $X+e-f \in \I_2$), 
    we obtain $q_2(e) < 0 \le q_2(f)$ by combining the above two cases.
    \item 
    For $e \in E \setminus X$ and $f \in X$ such that $X + e \not\in \I_2$ and $X + e - f \in \I_2$,
    we obtain $p(e) \le p(f) + \varepsilon$
    by Observation~\ref{obs:fundamentalcircuit} and Lemma~\ref{lem:propertyrank2}, implying that  
    $q_2(e) = p(e) - \varepsilon \le p(f) = q_2(f)$.
\end{itemize}
These observations show that $X$ is a $q_2$-maximum independent set of $\M_2$ by Lemma~\ref{lem:charaminind}. 
Since $\OPT \in \I_2$, we obtain $q_2(X) \ge q_2(\OPT)$.

Therefore, 
\begin{align*}
|X| &= \sum_{e \in X} (1 - p(e)) + \sum_{e \in X} p(e) \\
&=  q_1(X) + q_2(X) \\
&\ge q_1(\OPT) + q_2 (\OPT) \\
&\ge \sum_{e \in {\OPT}} (1 - p(e)) + \sum_{e \in {\OPT}} (p(e) - \varepsilon) \\
&= (1- \varepsilon) \opt, 
\end{align*}
which completes the proof for Theorem~\ref{thm:mainrank}. 
\end{proof}

\subsection{Relation to Generalized Stable Matching}
\label{sec:relationtoGSM}

In Algorithm~\ref{alg:ranklinear}, Alice repeatedly makes proposals and Bob accepts or rejects them, which resembles the behavior of the Gale--Shapley algorithm for the stable matching problem~\cite{GaleShapley1962}.
In this subsection, we show that the output of the algorithm is in fact related to a matroidal generalization of the stable matching problem.
Although this connection does not directly lead to an improvement in the algorithmic performance, it reveals an intriguing relationship between seemingly unrelated problems and is therefore worth mentioning.

Consider a market consisting of two agents, 1 and 2, and a set of contracts $\tilde E$.
Each agent $i \in \{1,2\}$ is associated with a matroid $\tilde \M_i = (\tilde E, \tilde \I_i)$ and a preference relation $\succ_i$ over $\tilde E$, where ties are allowed.
We say that a contract set $X \subseteq \tilde E$ is {\em feasible} if $X \in \tilde \I_1 \cap \tilde \I_2$.
For a feasible set $X \subseteq \tilde E$, a contract $e \in \tilde E \setminus X$ is {\em blocking} if, for each $i \in \{1,2\}$, either
(i) $X + e \in \tilde \I_i$, or
(ii) there exists $f \in X$ such that $e \succ_i f$ and $X + e - f \in \tilde \I_i$.
We say that a feasible set $X \subseteq \tilde E$ is {\em stable} if it admits no blocking contract.

It was shown by Fleiner~\cite{Fleiner03} that such a market (as well as more generalized markets) always admits a stable set. Furthermore, a stable set can be computed efficiently via a generalization of the Gale--Shapley algorithm for the stable matching problem~\cite{GaleShapley1962}. 
This algorithm serves as the main inspiration for our algorithms. 
Note that although only the case without ties was discussed in~\cite{Fleiner03}, the argument can also be applied to the case with ties after breaking them arbitrarily.

Let $\M_1=(E, \I_1)$ and $\M_2 = (E, \I_2)$ be an instance of the matroid intersection problem, 
and let $\varepsilon > 0$. 
Then, for each $i \in \{1, 2\}$, we define $\tilde \M_i = (\tilde E, \tilde \I_i)$ by   
\begin{align*}
\tilde E &= \{(e, \alpha) \mid e \in E,\ \alpha \in \{0, {1}/{\varepsilon}, {2}/{\varepsilon}, \dots, 1\} \}, \\ 
\tilde \I_i &= \{ \{(e_1, \alpha_1), \dots , (e_k, \alpha_k) \} \mid (e_i, \alpha_i) \in \tilde E,\ e_1, \dots , e_k \mbox{ are distinct, } \{e_1, \dots , e_k\} \in \I_i \}. 
\end{align*}
It is not difficult to see that each $\tilde \M_i$ is a matroid, since $\M_i$ is a matroid. 
We define preference relations $\succ_1$ and $\succ_2$ as follows: 
for $(e, \alpha), (e', \alpha') \in \tilde E$ (possibly, $e = e'$), 
$(e, \alpha) \succ_1 (e', \alpha')$ if and only if $\alpha < \alpha'$, and 
$(e, \alpha) \succ_2 (e', \alpha')$ if and only if $\alpha > \alpha'$. 
This defines a market $(\tilde E, \tilde \M_1, \tilde \M_2, \succ_1, \succ_2)$. 
Note that a similar construction was previously used in a different context~\cite{CKY24,CKTY25}. 
Then the output of Algorithm~\ref{alg:ranklinear} induces a stable set of contractions as follows. 

\begin{proposition}
Let $X$ be the output of Algorithm~\ref{alg:ranklinear}, 
and let $p$ be the price vector at the end of the algorithm. 
Then $\tilde X := \{(e, p(e)) \mid e \in X\}$ is a stable set in the market $(\tilde E, \tilde \M_1, \tilde \M_2, \succ_1, \succ_2)$.
\end{proposition}

\begin{proof}
Since $X \in \I_1 \cap \I_2$, it is obvious that $\tilde X \in \tilde \I_1 \cap \tilde \I_2$, i.e., $\tilde{X}$ is feasible. 
To prove that $\tilde X$ is a stable set, 
we show that any $(e, \alpha) \in \tilde E \setminus \tilde X$ is not a blocking contract.

If $\alpha \ge p(e)$, then we prove that (i) $\tilde X + (e, \alpha) \not\in \tilde \I_1$ and 
(ii) there exists no $(e', p(e')) \in \tilde X$ such that $(e, \alpha) \succ_1 (e', p(e'))$ and $\tilde X + (e, \alpha) - (e', p(e')) \in \tilde \I_1$ as follows. 
When $e \in X$, condition (i) is clear from the definition of $\tilde I_1$. 
When $e \in E \setminus X$, condition (i) follows immediately from $p(e) \le \alpha \le 1$ and the stopping condition of the algorithm.
If there exists $(e', p(e')) \in \tilde X$ that violates condition (ii), 
then we would have $p(e') > \alpha \ge p(e)$ (which implies $e' \neq e$) and $X + e - e' \in \I_1$, contradicting the fact that $X$ is a $q_1$-maximum independent set, as shown in the proof of Theorem~\ref{thm:mainrank}.

If $\alpha \le p(e) - \varepsilon$, then we prove that 
(i) $\tilde X + (e, \alpha) \not\in \tilde \I_2$ and 
(ii) there exists no $(e', p(e')) \in \tilde X$ such that $(e, \alpha) \succ_2 (e', p(e'))$ and $X + (e, \alpha) - (e', p(e')) \in \tilde \I_2$ as follows.  
When $e \in X$, condition (i) is clear, and condition (ii) holds because
$(e, \alpha) \not\succ_2 (e', p(e'))$ when $e' = e$, and
$X + (e, \alpha) - (e', p(e')) \not\in \tilde \I_2$ for $e' \in X - e$.
When $e \in E \setminus X$, Lemma~\ref{lem:propertyrank2} shows conditions (i) and (ii). 

Therefore, $(e, \alpha)$ is not a blocking contract, and hence $\tilde X$ is a stable set. 
\end{proof}

\section{Weighted Problem with Nearly Linear Rank Queries}
\label{sec:rankweighted}

In this section, to prove Theorem~\ref{thm:mainrankweighted}, 
we show that Algorithm~\ref{alg:ranklinear} extends naturally to the weighted matroid intersection problem.
The resulting algorithm is given in Algorithm~\ref{alg:ranklinearweight}.
The algorithm is very similar to Algorithm~\ref{alg:ranklinear}; the main difference is that the weight $w(e)$ is interpreted as Alice's valuation of the item $e \in E$.
Accordingly, we introduce a new variable $q(e) := w(e)-p(e)$, which represents Alice's utility for $e \in E$ at price $p(e)$.

The modifications to Algorithm~\ref{alg:ranklinear} are straightforward.
In line~6 of Algorithm~\ref{alg:ranklinearweight}, an element $e$ is chosen to maximize $q(e)$ rather than minimize $p(e)$.
Moreover, in lines~13 and~14, the price of an element $f$ is increased by $\varepsilon w(f)$ instead of $\varepsilon$.

Then, by the same argument as Observation~\ref{obs:easyobsrank} and Lemmas~\ref{lem:numberofrankqueries}, \ref{lem:propertyrank1}, and~\ref{lem:propertyrank2}, we obtain the following.

\begin{algorithm}
    \caption{Weighted Matroid Intersection with Nearly Linear Rank Queries}\label{alg:ranklinearweight}
    \KwInput{Two matroids $\M_1 = (E, \I_1)$ and $\M_2 = (E, \I_2)$, and weight $w(e)\in \mathbb{R}_+$ for $e \in E$}
    \KwOutput{A common independent set $X \in \I_1 \cap \I_2$} 
    $X \gets \emptyset$ ; \\
	\ForEach{$e \in E$}{       
    $p(e) \gets 0$ ; \\
    $q(e) \gets w(e)$ ; \\
    }
    \While {true}{
    Find $e \in E \setminus X$ maximizing $q(e)$ subject to $X + e \in \I_1$ ; \\
    \tcp{Let $e = \bot$ if no such element exists}
    \If{$e \neq \bot$ and $q(e) \ge 0$}{
        \If{$X + e \in \I_2$}{
            $X \gets X + e$ ; 
        }
        \Else{
        Find $f \in C_{2}(X \mid e)$ minimizing $p(f)$   \tcp*{possibly, $f=e$}
        $X \gets X + e - f$ ; \\                          
        $p(f) \gets p(f) + \varepsilon w(f)$ ; \\
        $q(f) \gets q(f) - \varepsilon w(f)$ ; 
        }
    }
    \Else{
        \Return $X$ ;         
    }
    }
\end{algorithm}

\begin{observation}\label{obs:easyobsrankweight}
    Throughout the execution of Algorithm~\ref{alg:ranklinearweight}, we have the following properties: 
    \begin{enumerate}
    \item 
    $X \in \I_1 \cap \I_2$.
    \item 
    $|X|$ is monotonically non-decreasing. 
    \item 
    $0\le p(e) \le w(e)$ for each $e \in X$. 
    \item 
    $0\le p(e) \le (1+\varepsilon) w(e)$ for each $e \in E$. 
    \item 
    For each $e \in E$, $p(e)$ is monotonically non-decreasing.
\end{enumerate}
\end{observation}

\begin{lemma}\label{lem:numberofrankqueriesweight}
    Algorithm~\ref{alg:ranklinearweight} uses $O({n \log n}/{\varepsilon})$ rank queries.
\end{lemma}

\begin{lemma}\label{lem:propertyrank1weight}
At the beginning of each iteration of the while-loop in Algorithm~\ref{alg:ranklinearweight}, 
$X$ is contained in a $q$-maximum base of $\M_1$.  
\end{lemma}

\begin{lemma}
\label{lem:propertyrank2weight}
At the beginning of each iteration of the while-loop in Algorithm~\ref{alg:ranklinearweight}, 
for any $e' \in E \setminus X$ with $p(e') \ge \varepsilon w(e')$, 
we have that 
\begin{itemize}
\item 
$X+e' \not\in \I_2$, and 
\item 
$p(e') \le p(f') + \varepsilon w(e')$ for any $f' \in C_2(X \mid e')$.
\end{itemize}
\end{lemma}

By using these lemmas, we can naturally extend Theorem~\ref{thm:mainrank} to Theorem~\ref{thm:mainrankweighted}, 
which we restate here. 

\mainrankweighted*

\begin{proof}
By Observation~\ref{obs:easyobsrankweight} and Lemma~\ref{lem:numberofrankqueriesweight}, 
Algorithm~\ref{alg:ranklinearweight} uses $O({n \log n}/{\varepsilon})$ rank queries and returns a common independent set $X \in \I_1 \cap \I_2$. 
Therefore, it suffices to show that the output $X$ of 
Algorithm~\ref{alg:ranklinearweight} satisfies $w(X) \ge (1-\varepsilon) w(\OPT_w)$, 
where $\OPT_{w}$ denotes a common independent set of maximum weight. 

By Lemma~\ref{lem:propertyrank1weight}, 
there exists a $q$-maximum base $\hat X$ of $\M_1$ that contains $X$. 
By Observation~\ref{obs:easyobsrankweight} and by the stopping condition of the algorithm, 
we have that 
\[
X = \{e \in \hat X \mid q(e) \ge 0\}. 
\]
Therefore, $X$ is a $q$-maximum independent set of $\M_1$, which implies that $q(X) \ge q(\OPT_w)$. 

Let
\[
p'(e) = 
\begin{cases}
p(e) & \mbox{for $e \in X$,} \\
p(e) - \varepsilon w(e) & \mbox{for $e \in E \setminus X$.}
\end{cases}
\]
We now observe the following: 
\begin{itemize}
    \item 
    For $e \in X$, we obtain $p'(e) = p(e) \ge 0$. 
    \item 
    For $e \in E \setminus X$ with $X + e \in \I_2$, 
    we obtain $p(e) < \varepsilon w(e)$ by Lemma~\ref{lem:propertyrank2weight}, 
    implying that $p'(e) < 0$. 
    \item 
    For $e \in E \setminus X$ and $f \in X$ such that $X + e \in \I_2$ (which implies that $X+e-f \in \I_2$), 
    we obtain $p'(e) < 0 \le p'(f)$ by combining the above two cases.
    \item 
    For $e \in E \setminus X$ and $f \in X$ such that $X + e \not\in \I_2$ and $X + e - f \in \I_2$,
    we obtain $p(e) \le p(f) + \varepsilon w(e)$
    by Observation~\ref{obs:fundamentalcircuit} and Lemma~\ref{lem:propertyrank2weight}, implying that  
    $p'(e) = p(e) - \varepsilon w(e) \le p(f) = p'(f)$.
\end{itemize}
These observations show that $X$ is a $p'$-maximum independent set of $\M_2$ by Lemma~\ref{lem:charaminind}. 
Since $\OPT_w \in \I_2$, we obtain $p'(X) \ge p'(\OPT_w)$. 

Therefore, 
\begin{align*}
w(X) &= \sum_{e \in X} (w(e) - p(e)) + \sum_{e \in X} p(e) \\
&=  q(X) + p'(X) \\
&\ge q(\OPT_w) + p' (\OPT_w) \\
&\ge \sum_{e \in {\OPT_w}} (w(e) - p(e)) + \sum_{e \in {\OPT_w}} (p(e) - \varepsilon w(e)) \\
&= (1- \varepsilon) w(\OPT_w), 
\end{align*}
which completes the proof. 
\end{proof}

\section{Semi-Streaming Algorithm}
\label{sec:semistreaming}

\subsection{Algorithm Description}

In this section, we present another auction-based algorithm and prove Theorem~\ref{thm:mainstreaming}. 
In the algorithm, we maintain a common independent set $X \in \I_1 \cap \I_2$ and a nonnegative value
$\pi(e) \in \mathbb{R}_{+}$ for each $e \in X$. 
For notational convenience, let $\pi(e) = \bot$ for $e \in E \setminus X$. 
Then, we need $O(|X|)$ space
to keep the information of $X$ and $\pi$, where we note that $|X| \le \opt$ 
and each value $\pi(e)$ is represented with $O(1)$ space as we will see in Observation~\ref{obs:pricerange}\footnote{More precisely, Observation~\ref{obs:pricerange} shows that each $\pi(e)$ can be represented using $O(\log (1/\varepsilon))$ bits, 
which is $O(\log n)$ bits because we have assumed that $1/\varepsilon \le n+1$.
If we measure space in terms of the number of words, where each word consists of $O(\log n)$ bits, then
we can say that $O(|X|)$ space is required to store the information of $X$ and $\pi$.
}.

For $X$ and $\pi$, we define the {\em price vector $p \in \mathbb{R}_+^E$ induced by $X$ and $\pi$} as follows: 
\[
p(e) = 
\begin{cases}
\pi(e) & \mbox{if $e \in X$,} \\
\varepsilon & \mbox{if $e \in E \setminus X$ and $X + e \in \I_2$,} \\
\min \{ \pi(f) \mid f \in C_{2}(X \mid e) - e \} + \varepsilon & \mbox{if $e \in E \setminus X$ and $X + e \not\in \I_2$.}
\end{cases}
\] 
Note that, 
for $e \in E \setminus X$ with $X + e \not\in \I_2$, 
there always exists an element in $C_{2}(X \mid e) - e$, 
because we have assumed that $\M_2$ has no loop. 
Throughout the section, 
we omit ``induced by $X$ and $\pi$'' if no confusion may arise.  
The algorithm does not maintain the value of $p(e)$ for each $e \in E$ explicitly. 
It only maintains $X$ and $\pi$, and computes the value of $p(e)$ for $e \in E$ whenever it is needed. 
Note that it can be computed with $O(|X|)$ space in a straightforward way (see Lemma~\ref{lem:computeprice} for an efficient algorithm).

Our algorithm is presented in Algorithm~\ref{alg:matroidint1}, and its implementation details will be refined in Section~\ref{sec:implement} later.
The proposed algorithm combines Algorithm~\ref{alg:ranklinear}, inspired by Fleiner's generalized stable matching framework~\cite{Fleiner03}, 
and a semi-streaming algorithm for the bipartite matching problem~\cite{assadi2021auction}.  
It proceeds similarly to Algorithm~\ref{alg:ranklinear} in that it can be viewed as an auction mechanism between Alice and Bob. 
Indeed, $X \in \I_1 \cap \I_2$ represents the items tentatively agreed upon for trade between Alice and Bob, and the price $p$ increases monotonically; see Lemma~\ref{lem:pricemonotone}.
However, there are a few key differences.  

The first difference is that we maintain only the price $\pi(e)$ for each element $e \in X$ in order to reduce space complexity.  
For each $e \in E \setminus X$, $p(e)$ is defined as the price that motivates Bob to sell $e$, possibly by removing some element from $X$.  

The second difference lies in the iteration step (Lines~3--14), which we call a {\em round}.  
In each round, let $X^*$ and $p^*$ denote the set of items and the price vector at the beginning of the round. 
Then Alice aims to purchase additional items at prices $p^*$.
Specifically, Alice proposes to Bob that $X^*$ be replaced with a superset $Y \supseteq X^*$, rather than adding a single element.  
Bob then removes some elements so that the resulting set $X$ belongs to $\I_2$.  
In fact, this process is executed sequentially as elements of $E \setminus X^*$ arrive one by one.

The third difference is that the algorithm repeats this process only ${2}/{\varepsilon^2}$ times, that is, a constant number of rounds, 
whereas Algorithm~\ref{alg:ranklinear} repeats the iterations exhaustively.

\begin{algorithm}
    \caption{Semi-Streaming Algorithm for Matroid Intersection}\label{alg:matroidint1}
    \KwInput{Two matroids $\M_1 = (E, \I_1)$ and $\M_2 = (E, \I_2)$}
    \KwOutput{A common independent set $X \in \I_1 \cap \I_2$} 
    Let $X \gets \emptyset$ ; \\
    \For{$t \leftarrow 1$ \KwTo ${2}/{\varepsilon^2}$}{    
  	$X^*, Y \gets X$ ; \\
  	$\pi^*(e) \gets \pi(e)$ for each $e \in E$ ; \\ 
	\tcp{We keep explicitly $\pi^*(e)$ only for $e \in X^*$, as $\pi^*(e) = \bot$ for $e \in E \setminus X^*$}
	\tcp{Let $p^*$ denote the price vector induced by $X^*$ and $\pi^*$}
    \ForEach{$e \in E \setminus X^*$}{       
          \If{$p(e) = p^*(e) \le 1$ and $Y + e$ is contained in a $p^*$-minimum base of $\M_1$} {
	           $Y \gets Y + e$ ; \\ 
	           $\pi(e) \gets p(e)$ ; \\ 
	           \If{$X+e \in \I_2$}{
	               $X \gets X + e$ ; \\ 
            	}
	           \Else{
            	   Find $f \in C_{2}(X \mid e) - e$ minimizing $\pi(f)$ ; \\
                  $X \gets X + e - f$ ; \\              
	               $\pi(f) \gets \bot$ ; \\
            	}
        	\tcp{The price vector $p$ induced by $X$ and $\pi$ is updated implicitly}
        }
    }
    }
    \Return $X$ ;
\end{algorithm}

We first show that Algorithm~\ref{alg:matroidint1} returns a common independent set $X \in \I_1 \cap \I_2$ whose size is close to $\opt$. 
We then discuss the query complexity and the number of streaming passes in Section~\ref{sec:implement} after refining the implementation of the algorithm.

\subsection{Useful Properties}

To prove that 
the size of the output $X$ of Algorithm~\ref{alg:matroidint1} is close to $\opt$,  
we show several properties that are maintained throughout the execution of the algorithm.

\begin{observation}\label{obs:pricerange}
Throughout the execution of Algorithm~\ref{alg:matroidint1}, we have the following properties:
\begin{enumerate}
    \item $X \in \I_2$, 
    \item $|X|$ is monotonically non-decreasing. 
    \item For each $e \in X$, 
    $\pi(e) = i \varepsilon$ for some $i \in \{1, 2, \dots , \frac{1}{\varepsilon}\}$. 
    \item For each $e \in E$, 
    $p(e) = i \varepsilon$ for some $i \in \{1, 2, \dots , \frac{1}{\varepsilon} + 1\}$. 
\end{enumerate}
\end{observation}

\begin{proof}
    We see that Property 1 is preserved in Lines 10 and 13 by Observation~\ref{obs:fundamentalcircuit}. Property 2 is obvious. 
    At the beginning of the algorithm, $X=\emptyset$ and $p(e) = \varepsilon$ for each $e \in E$, since $\M_2$ has no loops, 
    which means that Properties 3 and 4 hold. 
    When $\pi(e)$ is updated in Lines 8 and 14, $\pi(e) := p(e) \le 1$, and hence Property 3 is preserved. 
    Then Property 3 implies Property 4 by the definition of $p$. 
\end{proof}

By the first property in this observation, $C_2(X \mid e)$ is well-defined in Line 12. 
Furthermore, in Line 12 there always exists an element $f \in C_{2}(X \mid e) - e$, 
because we have assumed that $\M_2$ has no loops.
Note that Line 12 is equivalent to finding $f \in C_{2}(X \mid e)$ that minimizes $p(f)$, by the definition of $p$. 

\begin{observation}\label{obs:42}
    Let $X \in \I_2$, $e \in E \setminus X$, and $f \in X$ be such that $X + e -f \in \I_2$. 
    Then, $p(e) \le \pi(f) + \varepsilon$.
\end{observation}

\begin{proof}
Since $X + e - f \in \I_2$, 
we obtain either $X + e \in \I_2$ or $f \in C_2(X \mid e) - e$ by Observation~\ref{obs:fundamentalcircuit}. 
In the former case, $p(e) = \varepsilon \le \pi(f) + \varepsilon$.
In the latter case, $p(e) \le \pi(f) + \varepsilon$ holds by the definition of $p(e)$.    
\end{proof}

\begin{lemma}\label{lem:pricemonotone}
For each $e' \in E$, $p(e')$ is monotonically non-decreasing. 
\end{lemma}

\begin{proof}
Suppose that Lines 7--14 are executed for $e \in E \setminus X^*$. 
In this procedure, suppose that 
the objects $X$, $\pi$, and $p$ are updated to 
$X'$, $\pi'$, and $p'$, respectively. 
For $e' \in E$, we prove $p' (e')  \ge p (e')$ by considering the following cases separately. 

\medskip

\noindent
\textbf{Case 1.} Suppose that $X' = X + e$. 
\begin{enumerate}
\item
If $e' \in X$, then $p' (e') = \pi'(e') = \pi(e') = p(e')$. 
\item
If $e' = e$, then $p' (e') = \pi'(e') = p(e')$.  
\item
If $e' \in E \setminus X'$ and $X' + e' \in \I_2$, 
then $p' (e') = \varepsilon = p(e')$, 
where we note that $X + e' \in \I_2$. 
\item
Suppose that $e' \in E \setminus X'$ and $X' + e' \not\in \I_2$. 
Let $f'$ be a minimizer of $\pi'(f')$ subject to $f' \in X'$ and $X' + e' - f' \in \I_2$ (i.e., $f' \in C_2(X' \mid e') - e'$). 
Then, $p' (e') = \pi' (f') + \varepsilon$. 
If $f'=e$, then we obtain $X + e'  \in \I_2$, and hence $p' (e') \ge \varepsilon = p(e')$. 
If $f' \neq e$, then we have that $f' \in X$ and $X + e' - f' \in \I_2$, and hence  
$p(e') \le \pi(f') + \varepsilon$ by Observation~\ref{obs:42}. 
Therefore, we obtain
\[
p' (e') = \pi'(f') + \varepsilon = \pi(f') + \varepsilon \ge p(e').
\]
\end{enumerate}

\noindent
\textbf{Case 2.} Suppose that $X' = X + e -f$, 
where $f$ is the element chosen in Line 12. 
Then, $p(e) = \pi(f) + \varepsilon$ by definition. 
Note that $X+ e \not\in \I_2$. 
\begin{enumerate}
\item
If $e' \in X - f$, then $p' (e') = \pi'(e') = \pi(e') = p(e')$. 
\item
If $e' = e$, then $p' (e') = \pi'(e') = p(e')$.  
\item
Suppose that $e' = f$. In this case, $X' + e' = X + e \not\in \I_2$. 
Let $f'$ be a minimizer of $\pi'(f')$ subject to $f' \in X'$ and $X' + e' - f' \in \I_2$ (i.e., $f' \in C_2(X' \mid e') - e'$). 
Then, $p' (e') = \pi' (f') + \varepsilon$. 
\begin{itemize}
    \item 
    If $f' = e$, then we obtain 
        \begin{align*}
        p' (e') &= \pi' (f') + \varepsilon = \pi' (e) + \varepsilon =  p(e) + \varepsilon \\ 
        &= \pi(f) + 2 \varepsilon = \pi(e') + 2 \varepsilon =  p(e') + 2 \varepsilon > p(e').
        \end{align*}
    \item 
    If $f'\neq e$, then we have $f' \in C_2(X \mid e) -e$, because $f' \in C_2(X' \mid e') = C_2(X \mid e)$. 
    Since $f$ minimizes $\pi(f)$ on $C_2(X \mid e) -e$, we obtain 
    $\pi(f) \le \pi(f')$, and hence  
    \[
        p' (e') = \pi' (f') + \varepsilon = \pi (f') + \varepsilon \ge  \pi(f) + \varepsilon = \pi(e') + \varepsilon 
        = p(e') + \varepsilon > p(e').
    \]
\end{itemize}
\item
Suppose that $e' \in (E \setminus X)-e$ and $X' + e' \in \I_2$. 
Since $|X' + e'| > |X|$, 
there exists $e'' \in (X' + e') \setminus X = \{e, e'\}$ such that $X + e'' \in \I_2$. 
Since $X + e \not\in \I_2$, we obtain $e'' = e'$ and $X + e' \in \I_2$. 
Therefore, $p' (e') = \varepsilon = p(e')$. 
\item
Suppose that $e' \in (E \setminus X)-e$ and $X' + e' \not\in \I_2$. 
Let $f'$ be a minimizer of $\pi'(f')$ subject to $\pi'(f') \in X'$ and $X' + e' - f' \in \I_2$ (i.e., $f' \in C_2(X' \mid e') - e'$). 
Then, $p' (e') = \pi' (f') + \varepsilon$. 
\begin{itemize}
\item
If $f'=e$, then we have $X + e' - f = X' + e' - f' \in \I_2$, and hence 
$\pi(f) + \varepsilon \ge p(e')$ by Observation~\ref{obs:42}. Therefore, we obtain
\[
p' (e') = \pi' (f') + \varepsilon  = \pi' (e) + \varepsilon =  p(e) + \varepsilon =  \pi(f)+2 \varepsilon \ge p(e') + \varepsilon > p(e'). 
\]

\item
If $f' \neq e$, then $e, f, e'$, and $f'$ are distinct elements, and 
$X + e - f + e' - f' = X' + e' - f' \in \I_2$.
By applying the simultaneous exchange property to $X$ and $X + e - f + e' - f'$ (in the truncation of $\M_2$ to rank $|X|$),  
we obtain (i) $X + e - f' \in \I_2$ and $X + e' - f \in \I_2$ or (ii) $X + e' - f' \in \I_2$. 
In the case (i), $X + e - f' \in \I_2$ implies $\pi(f') \ge \pi(f)$ by the choice of $f$, and 
$X + e' - f \in \I_2$ implies $\pi(f) + \varepsilon \ge p(e')$ by Observation~\ref{obs:42}. 
Therefore, we obtain 
\[
p' (e') = \pi' (f') + \varepsilon =  \pi(f') + \varepsilon \ge \pi(f) + \varepsilon \ge p(e'). 
\]
In the case (ii), since $X + e' - f' \in \I_2$ implies $\pi(f') + \varepsilon \ge p(e')$ by Observation~\ref{obs:42}, we obtain 
\[
p' (e') = \pi' (f') + \varepsilon =  \pi(f') + \varepsilon \ge p(e'). 
\]
\end{itemize}
\end{enumerate}
This completes the proof for $p' (e') \ge p(e')$. 
\end{proof}

\begin{lemma}\label{lem:minbase0}
The following invariants hold 
before and after each iteration of the inner for-loop.
\begin{enumerate}
    \item 
    $X$ is contained in a $p$-minimum base of $\M_1$. 
    \item 
    $Y$ is contained in a $p^*$-minimum base of $\M_1$. 
    \item 
    $X \cup X^* \subseteq Y$. 
    \item 
    $p(e) = \pi(e) = p^*(e)$ for any $e \in X$. 
\end{enumerate}
\end{lemma}

\begin{proof}
     Obviously, $X=\emptyset$ satisfies the first invariant at the beginning of the algorithm. 
     Consider a fixed round and suppose that the first invariant holds just before this round.   
     Then, since $p^*$ and $Y$ are set to $p$ and $X$, respectively at the beginning of this round, 
     before the first iteration of the inner for-loop in this round, 
     all the invariants hold.  

     By the update rules for $X$, $Y$, and $\pi$, the second, third, and fourth invariants are preserved in this round. 
     This implies that $X$ is contained in a $p^*$-minimum base of $\M_1$.
     Since $p(e) = p^*(e)$ for $e \in X$ by the fourth invariant and $p(e) \ge p^*(e)$ for all $e \in E$ by Lemma~\ref{lem:pricemonotone},  
     it follows from Lemma~\ref{lem:gsproperty} that $X$ is contained in a $p$-minimum base of $\M_1$.  
     Thus, the first invariant holds in this round.     
\end{proof}

The following lemma is derived in a similar way to Theorem~\ref{thm:mainrank}.

\begin{lemma}\label{lem:activerank}
Before and after each iteration of the inner for-loop, it holds that 
\[
r_1 ( \{e \in E \mid p(e) \le 1 \}) \ge (1-\varepsilon) \opt.
\]
\end{lemma}

\begin{proof}
Let $q_1(e) = 1 - p(e)$ for each $e \in E$ and let
\[
q_2(e) = \begin{cases}
p(e) & \mbox{for $e \in X$,} \\
p(e) - \varepsilon & \mbox{for $e \in E \setminus X$.}
\end{cases}
\]
By Lemma~\ref{lem:minbase0}, 
$X$ is contained in a $q_1$-maximum base of $\M_1$. 
Since $p(e) \le 1$, i.e., $q_1(e) \ge 0$ for each $e \in X$ by Observation~\ref{obs:pricerange}, 
this implies that 
$X$ is contained in a $q_1$-maximum independent set of $\M_1$, say $\hat X$. 
Note that $q_1(e) \ge 0$ for each $e \in \hat X$, that is, $\hat X \subseteq \{e \in E \mid p(e) \le 1 \}$. 
Since $\OPT \in \I_1$, we see that 
$q_1(\hat X) \ge q_1(\OPT)$. 

We next observe the following: 
\begin{itemize}
    \item 
    For $e \in X$, we obtain $q_2(e) = p(e) > 0$.
    \item 
    For $e \in E \setminus X$ with $X + e \in \I_2$, we obtain $q_2(e) = p(e) - \varepsilon = 0$.  
    \item 
    For any $e \in E \setminus X$ and $f \in X$ such that $X+e-f \in \I_2$, 
    we obtain $q_2(e) = p(e) - \varepsilon \le \pi(f) = p(f) = q_2(f)$ by Observation~\ref{obs:42}. 
\end{itemize}
These observations show that $X$ is a $q_2$-maximum independent set of $\M_2$ by Lemma~\ref{lem:charaminind}. 
Since $\OPT \in \I_2$, we obtain $q_2(X) \ge q_2(\OPT)$. 

Therefore, 
\begin{align*}
r_1 ( \{e \in E \mid p(e) \le 1 \}) 
&\ge |\hat X| \\ 
&\ge \sum_{e \in \hat X} (1 - p(e)) + \sum_{e \in X} p(e) \\
&=  q_1(\hat X) + q_2(X) \\
&\ge q_1(\OPT) + q_2 (\OPT) \\
&\ge \sum_{e \in {\OPT}} (1 - p(e)) + \sum_{e \in {\OPT}} (p(e) - \varepsilon) \\
&= (1- \varepsilon) \opt, 
\end{align*}
which completes the proof. 
\end{proof}

\subsection{Potential Functions and Approximation Ratio}

We now define two potential functions as follows: 
\begin{align*}
\Phi_{\rm item} &:= \pi (X), \\
\Phi_{\rm buyer} &:= \min \left\{ p(S) \mid S \in \I_1,\ |S| =  {\opt} \right\}. 
\end{align*}
Intuitively, $\Phi_{\rm item}$ is the total price of 
the items tentatively agreed upon for trade between the two players, 
and $\Phi_{\rm buyer}$ is the total price of a cheapest set of $\opt$ items in $\I_1$ (i.e., those that Alice wants to buy).
In the next two lemmas, we show how the increments of these potential functions can be related to the growth speed of $X$. Such a relation will be used later (Proposition~\ref{prop:streamingratio}) to  lower bound the size of the final $X$.

For $t \in \{0, 1, 2, \dots,  \frac{2}{\varepsilon^2}\}$, 
objects $X$, $\pi$, $Y$, $\Phi_{\rm item}$, and $\Phi_{\rm buyer}$ at the end of round $t$ are denoted by 
$X^t$, $\pi^t$, $Y^t$, $\Phi^t_{\rm item}$, and $\Phi^t_{\rm buyer}$, respectively. 
Here, {\it at the end of round 0} refers to 
{\it just before round 1}; for example, $X^0 = \emptyset$. 

\begin{lemma}\label{lem:potential1}
For $t \in \{1, 2, \dots , \frac{2}{\varepsilon^2} \}$, 
$\Phi^t_{\rm item} - \Phi^{t-1}_{\rm item} = \varepsilon (|Y^t| - |X^{t-1}|)$.
\end{lemma}

\begin{proof}
We see that $\Phi_{\rm item}$ increases by $\varepsilon$ in each execution of Line 7--14. 
Indeed, if $e$ is added to $X$ in Line 10, then 
$\pi(e)$ is set to $p(e) = \varepsilon$ in Line 8, 
and hence $\Phi_{\rm item}$ increases by $\varepsilon$. 
If $X$ is replaced with $X+e-f$ in Line 13, then $\pi(e)$ is set to $p(e) = \pi(f) + \varepsilon$ in Line 8
by the definition of $p$ and the choice of $f$, 
and hence $\Phi_{\rm item}$ increases by $\varepsilon$.

Since Lines 7--14 are executed for $e \in E$ in round $t$ if and only if $e \in Y^t \setminus X^{t-1}$, 
this part is executed $|Y^t| - |X^{t-1}|$ times in round $t$. 
Therefore, we obtain $\Phi^t_{\rm item} - \Phi^{t-1}_{\rm item} = \varepsilon (|Y^t| - |X^{t-1}|)$.
\end{proof}

\begin{lemma}\label{lem:potential2}
For $t \in \{1, 2, \dots , \frac{2}{\varepsilon^2} \}$, 
$\Phi^t_{\rm buyer} - \Phi^{t-1}_{\rm buyer} \ge \varepsilon ((1-\varepsilon)\opt - |Y^t|)$.
\end{lemma}

\begin{proof}
We fix $t$.  
Let $p^{t-1}$ and $p^{t}$ be the price vectors induced by $X^{t-1}$ and $\pi^{t-1}$, and by $X^{t}$ and $\pi^{t}$, respectively. 
For $i \in \mathbb{Z}$, 
let $E^{t}_{i} := \{ e \in E \mid p^{t}(e) = i \varepsilon \}$ and $E^{t}_{\le i} := \{ e \in E \mid p^{t}(e) \le i \varepsilon\}$. 
Note that $\{E^{t}_{i} \mid i \in \{1, 2, \dots , \frac{1}{\varepsilon} +1\}\}$ forms a partition of $E$ by Observation~\ref{obs:pricerange}. 
We define $E^{t-1}_{i}$ and $E^{t-1}_{\le i}$ in a similar way. 
We show the following claims.

\begin{claim} \label{clm:diffrank0}
For any integer $i \le \frac{1}{\varepsilon}$ and for any $e \in E^{t}_{i} \cap E^{t-1}_{i}$, 
$e \in {\rm span}_{\M_1} (E^{t-1}_{\le i-1} \cup (Y^t \cap E^{t-1}_{i}))$. 
\end{claim}

\begin{proof}[Proof of the claim]
Let $e \in E^{t}_{i} \cap E^{t-1}_{i}$.
If $e \in Y^t$, then the claim is obvious as 
$e \in Y^t \cap E^{t-1}_{i}$. 
Thus, it suffices to consider the case when $e \notin Y^t$.

In this case, $e$ does not satisfy the condition in Line 6 in round $t$. 
Moreover, since $e \in E^{t}_{i} \cap E^{t-1}_{i}$ implies that $p(e)$ does not increase in round $t$, we have $p(e) = p^*(e) = i\varepsilon$ when we check the condition in Line 6 for $e$ in this round. 
Therefore, at this time, $Y+e$ is not contained in a $p^*$-minimum base of $\M_1$. 
Since $Y \subseteq Y^t$, this implies that 
$Y^t+e$ is not contained in a $p^*$-minimum base of $\M_1$.

Let $Z$ be a $p^*$-minimum base of $\M_1$ 
that contains $Y^t$, which exists by Lemma~\ref{lem:minbase0}. 
Since $Y^t+e$ is not contained in a $p^*$-minimum base of $\M_1$, we have $e \not\in Z$, and hence the fundamental circuit $C_1(Z \mid e)$ exists. 
We make the following observations about $C_1(Z \mid e)$. 

\begin{itemize}
    \item 
    For any $f \in C_1(Z \mid e) - e$, we have $p^*(f) \le p^*(e)$; otherwise, by Observation~\ref{obs:fundamentalcircuit}, $Z + e - f$ would be a base of $\M_1$ such that $p^*(Z+e-f) < p^*(Z)$, contradicting the $p^*$-minimality of $Z$. 
    \item 
    For any $f \in C_1(Z \mid e) - e$ with $p^*(f) = p^*(e)$, we have $f \in Y^t$; otherwise, by Observation~\ref{obs:fundamentalcircuit}, 
    $Z + e - f$ would be a $p^*$-minimum base of $\M_1$ that contains $Y^t+e$, a contradiction. 
\end{itemize}
Since $p^* = p^{t-1}$ in round $t$, these observations show that 
\[
C_1(Z \mid e) - e \subseteq E^{t-1}_{\le i-1} \cup (Y^t \cap E^{t-1}_{i}). 
\]
Since $C_1(Z \mid e)$ is a circuit containing $e$, we have that
\[
e \in 
\mathrm{span}_{\M_1}( C_1(Z \mid e) - e) \subseteq \mathrm{span}_{\M_1}(E^{t-1}_{\le i-1} \cup (Y^t \cap E^{t-1}_{i})), 
\]
which completes the proof. 
\end{proof}

\begin{claim} \label{clm:diffrank}
For any integer $i \le \frac{1}{\varepsilon}$, 
it holds that 
$r_1(E^{t}_{\le i }) \le r_1 (E^{t-1}_{\le i-1}) + |Y^t \cap E^{t-1}_{i}|$. 
\end{claim}

\begin{proof}[Proof of the claim]
Since $p^{t}(e) \ge p^{t-1}(e)$ for each $e \in E$ by Lemma~\ref{lem:pricemonotone}, it holds that
\[
E^{t}_{\le i} \subseteq E^{t-1}_{\le i-1} \cup (E^{t}_{i} \cap E^{t-1}_{i}). 
\]
This together with Claim~\ref{clm:diffrank0} shows that
\[
E^{t}_{\le i} \subseteq {\rm span}_{\M_1} (E^{t-1}_{\le i-1} \cup (Y^t \cap E^{t-1}_{i})). 
\] 
Therefore, we obtain
\[
r_1(E^{t}_{\le i}) \le r_1 (E^{t-1}_{\le i-1} \cup (Y^t \cap E^{t-1}_{i})) \le r_1 (E^{t-1}_{\le i-1}) + |Y^t \cap E^{t-1}_{i}|, 
\] 
which completes the proof. 
\end{proof}

For $S \subseteq E$, let $\hat r (S) = \min \{r_1 (S), \opt\}$, that is, 
$\hat r$ is the rank function of the truncation of $\M_1$ to rank $\opt$. 
Let $S^* \in \I_1$ be a minimizer of $p^{t-1}(S^*)$ subject to $|S^*| = \opt$. 
Since $S^*$ is obtained by a greedy algorithm for the matroid with rank function $\hat r$, 
we obtain $|S^* \cap E^{t-1}_{\le i}| = \hat r (E^{t-1}_{\le i} )$ for each $i$.
This implies that 
$|S^* \cap E^{t-1}_{i}| = \hat r (E^{t-1}_{\le i} ) - \hat r (E^{t-1}_{\le i-1})$ for each $i$.
Therefore, 
\begin{align}
\Phi^{t-1}_{\rm buyer} &= p^{t-1}(S^*) \nonumber \\
&=  \sum_{i=1}^{\frac{1}{\varepsilon} +1} |S^* \cap E^{t-1}_{i}| i \varepsilon \nonumber \\
&=  \sum_{i=1}^{\frac{1}{\varepsilon} +1} (\hat r (E^{t-1}_{\le i} ) - \hat r (E^{t-1}_{\le i-1}) ) i \varepsilon \nonumber \\
&=  (1+\varepsilon) \hat r(E^{t-1}_{\le \frac{1}{\varepsilon} +1}) - \varepsilon \sum_{i=0}^{{1}/{\varepsilon}} \hat r (E^{t-1}_{\le i} ), \nonumber \\
&=  (1+\varepsilon) \hat r(E) - \varepsilon \sum_{i=1}^{{1}/{\varepsilon}} \hat r (E^{t-1}_{\le i} ), \label{eq:03}
\end{align}
where we used $E^{t-1}_{\le \frac{1}{\varepsilon} +1} = E$ and $E^{t-1}_{\le 0} = \emptyset$ in the last equality.  
Similarly, we obtain 
\begin{equation}
\Phi^{t}_{\rm buyer} = (1+\varepsilon) \hat r(E) - \varepsilon \sum_{i=1}^{{1}/{\varepsilon}} \hat r (E^{t}_{\le i} ). \label{eq:04}
\end{equation}

Let $i^*$ be the minimum index $i$ such that $r(E^{t-1}_{\le i}) \ge (1-\varepsilon)\opt$. 
Note that $i^* \le \frac{1}{\varepsilon}$ by Lemma~\ref{lem:activerank}. 
Note also that Lemma~\ref{lem:pricemonotone} implies $E^{t-1}_{\le i} \supseteq E^{t}_{\le i}$ for each $i$. 
Then, we obtain 
\begin{align*}
\Phi^{t}_{\rm buyer} - \Phi^{t-1}_{\rm buyer}
&= \varepsilon \sum_{i=1}^{{1}/{\varepsilon}} \left( \hat r (E^{t-1}_{\le i} ) - \hat r (E^{t}_{\le i} )  \right) 
&\mbox{(by (\ref{eq:03}) and (\ref{eq:04}))}\\
&\ge \varepsilon \sum_{i=1}^{i^*} \left( \hat r (E^{t-1}_{\le i} ) - \hat r (E^{t}_{\le i} )  \right) 
&\mbox{(by $i^* \le \frac{1}{\varepsilon}$ and $E^{t-1}_{\le i} \supseteq E^{t}_{\le i}$)}\\
&\ge \varepsilon \left( \hat r (E^{t-1}_{\le i^*} ) + \sum_{i=1}^{i^*-1}  r_1 (E^{t-1}_{\le i} ) - \sum_{i=1}^{i^*} r_1 (E^{t}_{\le i} )  \right) &\qquad \mbox{(by minimality of $i^*$)}\\
&\ge \varepsilon \left( (1-\varepsilon)\opt  + \sum_{i=1}^{i^*}  \left( r_1 (E^{t-1}_{\le i-1} ) -  r_1 (E^{t}_{\le i} ) \right) \right) & \mbox{(by $r_1 (E^{t-1}_{\le i^*}) \ge (1-\varepsilon) \opt$)}\\
&\ge \varepsilon \left( (1-\varepsilon)\opt - \sum_{i=1}^{i^*}  |Y^t \cap E^{t-1}_{i}| \right) 
&\mbox{(by Claim~\ref{clm:diffrank})}\\
&\ge \varepsilon ((1-\varepsilon)\opt - |Y^t|). &
\end{align*}
This completes the proof for Lemma~\ref{lem:potential2}. 
\end{proof}

By Lemmas~\ref{lem:potential1} and~\ref{lem:potential2}, 
we obtain the following proposition. 

\begin{proposition}\label{prop:streamingratio}
Algorithm~\ref{alg:matroidint1} returns a common independent set $X \in \I_1 \cap \I_2$ such that $|X| \ge (1-2\varepsilon) \opt$. 
\end{proposition}

\begin{proof}
The output $X$ is a common independent set by Observation~\ref{obs:pricerange} and Lemma~\ref{lem:minbase0}. 
Thus, it suffices to prove that $|X| \ge (1-2\varepsilon) \opt$.

Assume to the contrary that the output $X$ satisfies $|X| < (1-2\varepsilon) \opt$. 
Then, since $|X|$ is monotonically non-decreasing during the algorithm, we have $|X^t| < (1-2 \varepsilon) \opt$ 
for each round $t \in \{1, 2, \dots , \frac{2}{\varepsilon^2} \}$. 
By Lemmas~\ref{lem:potential1} and~\ref{lem:potential2}, we obtain
\[
(\Phi^t_{\rm item} + \Phi^t_{\rm buyer}) - (\Phi^{t-1}_{\rm item} + \Phi^{t-1}_{\rm buyer}) \ge \varepsilon ((1-\varepsilon) \opt - |X^{t-1}|) > \varepsilon^2  \opt 
\]
for each $t$. 
Since we execute ${2}/{\varepsilon^2}$ rounds in the algorithm, 
$\Phi_{\rm item} + \Phi_{\rm buyer}$ increases by more than $2 \opt$ in total. 
This contradicts the fact that 
$0 \le \Phi_{\rm item} \le \opt$ and $\varepsilon \opt  \le \Phi_{\rm buyer} \le (1+\varepsilon) \opt $. 
\end{proof}

Note that, by replacing $\varepsilon$ with ${\varepsilon}/{2}$, we obtain an algorithm that returns a set $X$ of size at least $(1 - \varepsilon) \opt$.

\subsection{Refined Implementation}
\label{sec:implement}

An issue in Algorithm~\ref{alg:matroidint1} is that it is not straightforward to check whether $Y+e$ is contained in a $p^*$-minimum base of $\M_1$ efficiently in Line 6; 
in particular, a naive approach would require an additional streaming pass.
To overcome this difficulty, in this subsection, we modify the algorithm to maintain a $p^*$-minimum base $Z$ of $\M_1$ and use it to check this condition. 
The refined algorithm is described in Algorithm~\ref{alg:matroidint2}. 

\begin{algorithm}
    \caption{Refined Implementation of Algorithm~\ref{alg:matroidint1}}\label{alg:matroidint2}
    \KwInput{Two matroids $\M_1 = (E, \I_1)$ and $\M_2 = (E, \I_2)$}
    \KwOutput{A common independent set $X \in \I_1 \cap \I_2$} 
    Let $X \gets \emptyset$ ; \\
    \For{$t \leftarrow 1$ \KwTo ${2} / {\varepsilon^2}$}{    
  	$X^*, Y \gets X$ ; \\
  	$\pi^*(e) \gets \pi(e)$ for each $e \in E$ ; \\ 
	\tcp{We keep explicitly $\pi^*(e)$ only for $e \in X^*$, as $\pi^*(e) = \bot$ for $e \in E \setminus X^*$}
	\tcp{Let $p^*$ denote the price vector induced by $X^*$ and $\pi^*$}
    Find a $p^*$-minimum base $Z$ of $\M_1$ that contains $Y$; \\
    \ForEach{$e \in E \setminus X^*$}{       
         \If{$p(e) = p^*(e) \le 1$ and $\exists e' \in Z \setminus Y$ s.t. $Z - e' + e  \in \B_1$ and $p^*(e') = p^*(e)$} {
                $Z \gets Z - e' + e$                \tcp*{possibly, $e'=e$} 
	           $Y \gets Y + e$ ; \\ 
	           $\pi(e) \gets p(e)$ ; \\ 
	           \If{$X+e \in \I_2$}{
	               $X \gets X + e$ ; \\ 
            	}
	           \Else{
            	   Find $f \in C_{2}(X \mid e) - e$ minimizing $\pi(f)$ ; \\
                  $X \gets X + e - f$ ; \\              
	               $\pi(f) \gets \bot$ ; \\
            	}
        	\tcp{The price vector $p$ induced by $X$ and $\pi$ is updated implicitly}
        }
    }
    }
    \Return $X$ ;
\end{algorithm}

The following lemma ensures that Algorithm~\ref{alg:matroidint2} behaves in the same way as Algorithm~\ref{alg:matroidint1}.

\begin{lemma}\label{lem:rephrasecondition}
    Suppose that $Y \subseteq E$ is contained in a $p^*$-minimum base $Z$ of $\M_1$, and let $e \in E \setminus Y$.  
    Then, the followings are equivalent: 
    \begin{enumerate}
    \item[(i)]    
    There exists $e' \in Z \setminus Y$ such that $Z - e' + e  \in \B_1$ and $p^*(e') = p^*(e)$. 
    \item[(ii)] 
    $Y + e$ is contained in a $p^*$-minimum base of $\M_1$. 
    \end{enumerate}
\end{lemma}

\begin{proof}
{\bf ((i) $\Rightarrow$ (ii))}
Since $Z$ is a $p^*$-minimum base of $\M_1$, 
if there exists $e' \in Z \setminus Y$ such that $Z - e' + e  \in \B_1$ and $p^*(e') = p^*(e)$, 
then $Z' := Z - e' + e$ is also a $p^*$-minimum base of $\M_1$. 
Since $Y + e$ is contained in $Z'$, (ii) holds. 

\noindent
{\bf ((ii) $\Rightarrow$ (i))}
If $e \in Z$, then (i) is obvious, because we can take $e' = e$. 
Suppose that $e \not\in Z$ and 
let $Z'$ be a $p^*$-minimum base of $\M_1$ that contains $Y+e$. 
By applying the simultaneous exchange property to $Z$ and $Z'$ with respect to $e$,  
there exists $e' \in Z \setminus Z'$ such that $Z-e'+e \in \B_1$ and $Z'+e'-e \in \B_1$. 
Since both $Z$ and $Z'$ are $p^*$-minimum bases of $\M_1$, 
$p^*(Z-e'+e) \ge p^*(Z)$ and $p^*(Z'+e'-e) \ge p^*(Z')$. 
Combining these inequalities, we obtain $p^*(e') = p^*(e)$. 
Since $e' \in Z \setminus Z' \subseteq Z \setminus Y$, $e'$ is a desired element in (i). 
\end{proof}

We note that in Line~5 of Algorithm~\ref{alg:matroidint2}, there always exists a $p^*$-minimum base $Z$ of $\M_1$ containing $Y$, by Lemma~\ref{lem:minbase0}.
Furthermore, the condition that $Z$ is a $p^*$-minimum base of $\M_1$ containing $Y$ is preserved when we update $Y$ and $Z$ in Lines~8 and~9. 
We also note that when Line~7 is executed for $e \in E \setminus X^*$, $e$ is not contained in $Y$, because $Y$ consists of $X^*$ and some elements for which the inner for-loop was executed before $e$ in this round. 
Therefore, by Lemma~\ref{lem:rephrasecondition}, Algorithm~\ref{alg:matroidint2} behaves in the same way as Algorithm~\ref{alg:matroidint1}.

\subsection{Query Complexity and Number of Passes}
\label{sec:complexity}

Finally in this subsection, we analyze the query complexity and 
number of passes in Algorithm~\ref{alg:matroidint2}. 
We first show that $p(e)$ can be computed efficiently.

\begin{lemma}\label{lem:computeprice}
Given $X$ and $\pi$, for $e \in E$,  
we can compute $p(e)$ using $O(\log \opt)$ independence queries. 
\end{lemma}

\begin{proof}
It suffices to consider the case when $X+e \not\in \I_2$. 
By applying Lemma~\ref{lem:binary} (1) with $S = T = X$, $v = e$, and $w = \pi$, 
we can find a minimizer $f \in C_2 (X \mid e) - e$ of $\pi(f)$
using $O(\log |X|) = O(\log \opt)$ independence queries to $\M_2$. 
Then, we obtain $p(e) = \pi(f) + \varepsilon$. 
\end{proof}

Similarly, given $X^*$ and $\pi^*$, we can compute $p^*(e)$ using $O(\log \opt)$ queries. 
To achieve a further minor improvement in the query complexity, 
the algorithm explicitly maintains $p^*(e)$ for each $e \in Z$, 
while $p^*(e)$ for $e \in E \setminus Z$ is computed from $X^*$ and $\pi^*$ only when needed. 
This requires $O(|Z|) = O(\opt)$ space. 
With this modification, in each round, 
the algorithm computes the value of $p^*(e)$ for each $e \in E$ exactly once.

The following lemmas show that Lines 5 and 7 can be executed efficiently. 

\begin{lemma}\label{lem:computeminbase}
    In Line 5 of Algorithm~\ref{alg:matroidint2}, 
    using one streaming pass and $O( n \log \opt )$ independence queries, 
    we can compute a $p^*$-minimum base $Z$ of $\M_1$ that contains $Y$.  
\end{lemma}

\begin{proof}
    We consider a one pass semi-streaming algorithm that updates an independent set $\M_1$ greedily. 
    The algorithm initializes $Z := Y$. 
    When an element in $Y$ arrives, we do nothing. 
    When an element $e \in E \setminus Y$ arrives, we compute $p^*(e)$ with $O( \log \opt )$ independence queries to $\M_2$ by Lemma~\ref{lem:computeprice}, and then do the following: 
    \begin{itemize}
    \item 
    If $Z + e \in \I_1$, then update $Z := Z + e$. 
    \item 
    If $Z + e \not\in \I_1$, then 
    find $f \in (Z + e) \setminus Y$ that maximizes $p^*(f)$ subject to $Z+e-f \in \I_1$, 
    and 
    update $Z := Z + e - f$ (possibly, $e=f$). 
    \end{itemize}
    The first case can be handled using $O(1)$ independence queries.  
    The second case can be implemented using $O( \log \opt )$ independence queries to $\M_1$ as follows: 
    we find $f' \in Z \setminus Y$ that maximizes $p^*(f')$ subject to $Z+e-f' \in \I_1$
    by applying Lemma~\ref{lem:binary} (1) with $S = Z$, $T = Z \setminus Y$, $v=e$, and $w = - p^*$, 
    and then compare $p^*(f')$ and $p^*(e)$. 
    Note that we explicitly maintain $p^*(e)$ for each $e \in Z$, as remarked before this lemma. 
    Therefore, this algorithm requires $O( n \log \opt )$ independence queries to $\M_1$ and $\M_2$ in total. 

    It remains to show that the output $Z^*$ of the algorithm is a $p^*$-minimum base of $\M_1$ containing $Y$. 
    We denote $E \setminus Y = \{e^1, \dots , e^\ell\}$ such that these elements arrive in this order.     
    For each $t \in \{0, 1, \dots , \ell\}$, 
    let $Z^t$ denote the set $Z$ at the end of the iteration for $e^t$, 
    let $E^t := Y \cup \{e^1, \dots , e^t\}$, 
    and let $\M^t_1 = (E^t, \I^t_1)$ denote the restriction of $\M_1$ to $E^t$, where 
    $\I^t_1 := \{S \in \I_1 \mid S \subseteq E^t\}$. 
    It is well known (and easy to show) that $\M^t_1$ is a matroid.  
    We inductively show that 
    $Z^t$ minimizes $p^*(Z^t)$ subject to being a base of the matroid $\M^t_1$ that contains $Y$. 
    
    The claim trivially holds when $t=0$. 
    Suppose that the claim holds for $t$, and consider the iteration for $e^{t+1}$. 
    Clearly, $Z^{t+1}$ is a base of $\M^{t+1}_1$ containing $Y$ by construction. 
    We show that it minimizes $p^*(Z^t)$ as follows. 
    \begin{itemize}
        \item 
        Suppose that $Z^{t+1} = Z^t + e^{t+1} \in \I_1$. 
        In this case, every base of $\M^{t+1}_1$ contains $e^{t+1}$.
        Therefore, for any base $Z'$ of $\M^{t+1}_1$ containing $Y$, 
        $Z' - e^{t+1}$ is a base of $\M^{t}_1$, and hence we obtain 
        \[
        p^*(Z^{t+1}) = p^*(Z^{t}) + p^*(e^{t+1}) \le p^*(Z' - e^{t+1}) + p^*(e^{t+1}) = p^*(Z'). 
        \]
        This shows the claim for $t+1$. 
        \item 
        Suppose that $Z^t + e^{t+1} \not\in \I_1$ and $Z^{t+1} = Z^t +e^{t+1}-f^{t+1}$, 
        where $f^{t+1} \in (Z^t + e^{t+1}) \setminus Y$ is the element chosen in the algorithm. 
        Let $Z'$ be an arbitrary base of $\M^{t+1}_1$ containing $Y$. 
        Our aim is to prove that $p^*(Z^{t+1}) \le p^*(Z')$.  

        If $e^{t+1} \not\in Z'$, then we have that $p^*(Z^{t}) \le p^*(Z')$, because the claim holds for $t$.
        Furthermore, we obtain $p^*(f^{t+1}) \ge p^*(e^{t+1})$ by the maximality of $p^*(f^{t+1})$, and hence
        \[
        p^*(Z^{t+1}) = p^*(Z^t +e^{t+1}-f^{t+1}) \le p^*(Z^{t}) \le p^*(Z'). 
        \]
        
        Otherwise, $e^{t+1} \in Z'$. 
        In this case, by the simultaneous exchange property for $Z^t$ and $Z'$ with respect to $e^{t+1} \in Z' \setminus Z^t$,  
        there exists $f' \in Z^t \setminus Z'$ such that 
        $W := Z^t + e^{t+1} - f' \in \I_1$ and $W' := Z' - e^{t+1} + f' \in \I_1$. 
        Since $Y \subseteq W' \subseteq E^t$ and the claim holds for $t$, we obtain $p^*(Z^t) \le p^*(W')$. 
        By combining this with $p^*(W) + p^*(W') = p^*(Z^t) + p^*(Z')$, 
        we have that $p^*(W) \le p^*(Z')$. 
        Furthermore, we obtain $p^*(f^{t+1}) \ge p^*(f')$ by the maximality of $p^*(f^{t+1})$, and hence 
        \[
        p^*(Z^{t+1}) = p^*(Z^t +e^{t+1}-f^{t+1}) \le p^*(W) \le p^*(Z'). 
        \]
    \end{itemize}

Therefore, the claim holds for every $t$. In particular,  
the output $Z$ minimizes $p^*(Z)$ subject to being a base of $\M_1$ that contains $Y$. 
Since $Y$ is contained in a $p^*$-minimum base of $\M_1$ by Lemma~\ref{lem:minbase0}, 
$Z$ is a $p^*$-minimum base of $\M_1$. 
\end{proof}

\begin{lemma}
    In Line 7 of Algorithm~\ref{alg:matroidint2}, for $e \in E \setminus Y$, we can find $e' \in Z \setminus Y$ such that $Z + e - e' \in \I_1$ and $p^*(e')=p^*(e)$ using 
    $O(\log \opt)$ independence queries 
    if such $e'$ exists. 
\end{lemma}

\begin{proof}
Recall that the algorithm explicitly maintains $p^*(e)$ for each $e \in Z$. 
We first compute $p^*(e)$ using $O( \log \opt )$ independence queries to $\M_2$ by Lemma~\ref{lem:computeprice}. 
We then find $e' \in Z \setminus Y$ that maximizes $p^*(e')$ subject to $Z + e - e' \in \I_1$.  
This can be done using $O( \log \opt )$ independence queries to $\M_2$ 
by applying Lemma~\ref{lem:binary} (1) with $S = Z$, $T = Z \setminus Y$, $v = e$, and $w = - p^*$. 
If $p^*(e') = p^*(e)$, then $e'$ is a desired element. 
If $p^*(e') < p^*(e)$, then we conclude that no desired element exists. 
Note that $p^*(e') \le p^*(e)$ holds, because $Z$ is a $p^*$-minimum base of $\M_1$. 
\end{proof}

By these lemmas, Algorithm~\ref{alg:matroidint2} can be implemented as a deterministic semi-streaming algorithm 
using $O({1}/{\varepsilon^2})$ passes, $O(\opt)$ space, and $O({n \log \opt}/{\varepsilon^2})$ independence queries. 
This together with Proposition~\ref{prop:streamingratio} proves Theorem~\ref{thm:mainstreaming}.

\section{Concluding Remarks}
\label{sec:conclusion}

In this paper, we proposed auction-based algorithms for the matroid intersection problem and obtained improved approximation algorithms in several computational settings. 
An important direction for future work is to further improve the number of passes, space complexity, and query complexity.

As shown in Corollary~\ref{cor:weighted}, Theorem~\ref{thm:mainstreaming} can be extended to the weighted setting via the reduction technique of Dudeja and Grilnberger~\cite{DG26}. 
A natural open question is whether Algorithm~\ref{alg:matroidint1} can be extended to the weighted setting more directly, leading to improved algorithms for weighted matroid intersection.

\bibliographystyle{abbrv}
\bibliography{ref}

@book{Schrijver2003,
  title={Combinatorial Optimization: Polyhedra and Efficiency},
  author={Schrijver, Alexander},
  year={2003},
  publisher={Springer}
}

@book{murota2003discrete,
  title={Discrete Convex Analysis},
  author={Murota, K.},
  isbn={9780898718508},
  lccn={2003042468},
  series={Discrete Mathematics and Applications},
  year={2003},
  publisher={Society for Industrial and Applied Mathematics}
}

@inproceedings{blikstad2021breaking,
  title={Breaking {$O(nr)$} for matroid intersection},
  author={Blikstad, Joakim},
  booktitle={Proceedings of the 48th International Colloquium on Automata, Languages, and Programming (ICALP 2021)},
  pages={31:1-31:17},
  year={2021},
  doi={10.4230/LIPIcs.ICALP.2021.31},
}

@inproceedings{AG11,
  author       = {Kook Jin Ahn and
                  Sudipto Guha},
  title        = {Linear Programming in the Semi-streaming Model with Application to
                  the Maximum Matching Problem},
  booktitle    = {Proceedings of the 38th International Colloquium on Automata, Languages, and Programming ({ICALP} 2011)},
  pages        = {526--538},
  year         = {2011},
}

@article{AG18, 
author = {Ahn, Kook Jin and Guha, Sudipto},
title = {Access to Data and Number of Iterations: Dual Primal Algorithms for Maximum Matching under Resource Constraints},
year = {2018},
volume = {4},
journal = {ACM Trans. Parallel Comput.},
pages = {Article 17},
}

@inproceedings{AJJST22,
  author       = {Sepehr Assadi and
                  Arun Jambulapati and
                  Yujia Jin and
                  Aaron Sidford and
                  Kevin Tian},
  title        = {Semi-Streaming Bipartite Matching in Fewer Passes and Optimal Space},
  booktitle    = {Proceedings of the 2022 {ACM-SIAM} Symposium on Discrete Algorithms (SODA 2022)},
  pages        = {627--669},
  publisher    = {{SIAM}},
  year         = {2022},
}

@article{EggertKMS12,
  author       = {Sebastian Eggert and
                  Lasse Kliemann and
                  Peter Munstermann and
                  Anand Srivastav},
  title        = {Bipartite Matching in the Semi-streaming Model},
  journal      = {Algorithmica},
  volume       = {63},
  number       = {1-2},
  pages        = {490--508},
  year         = {2012},
}

@article{SI95,
  author       = {Maiko Shigeno and
                  Satoru Iwata},
  title        = {A dual approximation approach to weighted matroid intersection},
  journal      = {Oper. Res. Lett.},
  volume       = {18},
  number       = {3},
  pages        = {153--156},
  year         = {1995},
}

@article{Garg23,
  author       = {Paritosh Garg and
                  Linus Jordan and
                  Ola Svensson},
  title        = {Semi-streaming algorithms for submodular matroid intersection},
  journal      = {Math. Program.},
  volume       = {197},
  number       = {2},
  pages        = {967--990},
  year         = {2023},
}

@article{FZ95,
  author       = {Satoru Fujishige and Xiaodong Zhang},
  title        = {An efficient cost scaling algorithm for the independent assignment problem},
  journal      = {Journal of the Operations Research Society of Japan},
  volume       = {38},
  number       = {1},
  pages        = {124--136},
  year         = {1995},
}

@inproceedings{BT25,
  author       = {Joakim Blikstad and
                  Ta{-}Wei Tu},
  title        = {Efficient Matroid Intersection via a Batch-Update Auction Algorithm},
  booktitle    = {Proceedings of the 8th Symposium on Simplicity in Algorithms (SOSA 2025)},
  pages        = {226--237},
  publisher    = {{SIAM}},
  year         = {2025},
}

@inproceedings{DG26,
  author       = {Aditi Dudeja and Mara Grilnberger},
  title        = {A Weighted-to-Unweighted Reduction for Matroid Intersection},
  booktitle    = {Proceedings of the 27th International Conference on Integer Programming and Combinatorial Optimization (IPCO 2026)},
  year         = {2026},
  pages        = {378--393}, 
  note         = {See also arXiv preprint arXiv:2602.15702}
}

@article{CalinescuCPV11,
  author       = {Gruia Calinescu and
                  Chandra Chekuri and
                  Martin P{\'{a}}l and
                  Jan Vondr{\'{a}}k},
  title        = {Maximizing a Monotone Submodular Function Subject to a Matroid Constraint},
  journal      = {{SIAM} J. Comput.},
  volume       = {40},
  number       = {6},
  pages        = {1740--1766},
  year         = {2011},
}

@article{FilmusW14,
  author       = {Yuval Filmus and
                  Justin Ward},
  title        = {Monotone Submodular Maximization over a Matroid via Non-Oblivious
                  Local Search},
  journal      = {{SIAM} J. Comput.},
  volume       = {43},
  number       = {2},
  pages        = {514--542},
  year         = {2014},
}

@article{ChakrabartiK15,
  author       = {Amit Chakrabarti and
                  Sagar Kale},
  title        = {Submodular maximization meets streaming: matchings, matroids, and
                  more},
  journal      = {Math. Program.},
  volume       = {154},
  number       = {1-2},
  pages        = {225--247},
  year         = {2015},
}

@inproceedings{Quanrud24,
  author       = {Kent Quanrud},
  title        = {Adaptive Sparsification for Matroid Intersection},
  booktitle    = {Proceedings of the 51st International Colloquium on Automata, Languages, and Programming
                  ({ICALP} 2024)},
  pages        = {118:1--118:20},
  year         = {2024},
}

@inproceedings{ChekuriGQ15,
  author       = {Chandra Chekuri and
                  Shalmoli Gupta and
                  Kent Quanrud},
  title        = {Streaming Algorithms for Submodular Function Maximization},
  booktitle    = {Proceedings of the 42nd International Colloquium on Automata, Languages, and Programming (ICALP 2015)},
  pages        = {318--330}, 
  year         = {2015},
}

@inproceedings{HuangTW20,
  author =	{Huang, Chien-Chung and Thiery, Theophile and Ward, Justin},
  title =	{Improved Multi-Pass Streaming Algorithms for Submodular Maximization with Matroid Constraints},
  booktitle =	{Approximation, Randomization, and Combinatorial Optimization. Algorithms and Techniques (APPROX/RANDOM 2020)},
  pages =	{62:1--62:19},
  year =	{2020},
  doi =		{10.4230/LIPIcs.APPROX/RANDOM.2020.62},
}

@article{FeldmanLNSZ26,
  author       = {Moran Feldman and
                  Paul Liu and
                  Ashkan Norouzi{-}Fard and
                  Ola Svensson and
                  Rico Zenklusen},
  title        = {Streaming Submodular Maximization Under Matroid Constraints},
  journal      = {Math. Oper. Res.},
  volume       = {51},
  number       = {1},
  pages        = {299--332},
  year         = {2026},
}

@inproceedings{blikstad2021breaking_STOC,
  title={Breaking the quadratic barrier for matroid intersection},
  author={Blikstad, Joakim and van den Brand, Jan and Mukhopadhyay, Sagnik and Nanongkai, Danupon},
  booktitle={Proceedings of the 53rd Annual ACM SIGACT Symposium on Theory of Computing (STOC 2021)},
  pages={421--432},
  year={2021},
  doi={10.1145/3406325.3451092},
}

@inproceedings{chakrabarty2019faster,
  title={Faster matroid intersection},
  author={Chakrabarty, Deeparnab and Lee, Yin Tat and Sidford, Aaron and Singla, Sahil and Wong, Sam Chiu-wai},
  booktitle={Proceedings of the 60th Annual Symposium on Foundations of Computer Science (FOCS 2019)},
  pages={1146--1168},
  year={2019},
  organization={IEEE},
  doi={10.1109/FOCS.2019.00072},
}

@inproceedings{chekuri2016fast,
  title={A fast approximation for maximum weight matroid intersection},
  author={Chekuri, Chandra and Quanrud, Kent},
  booktitle={Proceedings of the 27th Annual ACM-SIAM Symposium on Discrete Algorithms (SODA 2016)},
  pages={445--457},
  year={2016},
  organization={SIAM},
  doi={10.1137/1.9781611974331.ch33},
}

@article{cunningham1986improved,
  title={Improved bounds for matroid partition and intersection algorithms},
  author={Cunningham, William H},
  journal={SIAM Journal on Computing},
  volume={15},
  number={4},
  pages={948--957},
  year={1986},
  publisher={SIAM},
  doi={10.1137/0215066},
}

@article{edmonds1979matroid,
  title={Matroid intersection},
  author={Edmonds, Jack},
  journal={Annals of Discrete Mathematics},
  volume={4},
  pages={39--49},
  year={1979},
  publisher={Elsevier},
  doi={10.1016/S0167-5060(08)70817-3},
}

@incollection{edmonds1970matroid,
  title={Submodular functions, matroids, and certain polyhedra},
  author={Edmonds, Jack},
  booktitle={Combinatorial structures and their applications},
  pages={69--87},
  year={1970},
}

@article{huang2016exact,
  title={Exact and approximation algorithms for weighted matroid intersection},
  author={Huang, Chien-Chung and Kakimura, Naonori and Kamiyama, Naoyuki},
  journal={Mathematical Programming},
  volume={177},
  number={1-2},
  pages={85--112},
  year={2019},
  publisher={Springer},
  doi={10.1007/s10107-018-1260-x},
}

@article{lawler1975matroid,
  title={Matroid intersection algorithms},
  author={Lawler, Eugene L},
  journal={Mathematical Programming},
  volume={9},
  number={1},
  pages={31--56},
  year={1975},
  publisher={Springer},
  doi={10.1007/BF01681329},
}

@misc{nguyen2019note,
  title={A note on {C}unningham's algorithm for matroid intersection},
  author={Nguy$\tilde{{\hat{\text{e}}}}$n, Huy L},
  howpublished={arXiv preprint arXiv:1904.04129},
  year={2019}
}

@article{assadi2024simple,
  author       = {Sepehr Assadi},
  title        = {A Simple ($1-\varepsilon$)-Approximation Semi-Streaming Algorithm
                  for Maximum (Weighted) Matching},
  journal      = {TheoretiCS},
  volume       = {4},
  pages          = {Article 16},
  year         = {2025},
  doi          = {10.46298/THEORETICS.25.16},
}

@inproceedings{assadi2021auction,
  title={An auction algorithm for bipartite matching in streaming and massively parallel computation models},
  author={Assadi, Sepehr and Liu, S Cliff and Tarjan, Robert E},
  booktitle={Proceedings of the 4th Symposium on Simplicity in Algorithms (SOSA 2021)},
  pages={165--171},
  year={2021},
  organization={SIAM}
}

@InProceedings{terao:LIPIcs.WADS.2025.50,
  author =	{Terao, Tatsuya},
  title =	{Deterministic $(2/3 - \varepsilon)$-Approximation of Matroid Intersection Using Nearly-Linear Independence-Oracle Queries},
  booktitle =	{Proceedings of the 19th International Symposium on Algorithms and Data Structures (WADS 2025)},
  pages =	{50:1--50:18},
  year =	{2025},
  doi =		{10.4230/LIPIcs.WADS.2025.50},
}

@article{Fleiner03,
  author       = {Tam{\'{a}}s Fleiner},
  title        = {A Fixed-Point Approach to Stable Matchings and Some Applications},
  journal      = {Math. Oper. Res.},
  volume       = {28},
  number       = {1},
  pages        = {103--126},
  year         = {2003},
  doi          = {10.1287/MOOR.28.1.103.14256},
}

@article{GaleShapley1962,
author = {D. Gale and L. S. Shapley},
title = {College Admissions and the Stability of Marriage},
journal = {The American Mathematical Monthly},
volume = {69},
number = {1},
pages = {9--15},
year = {1962},
publisher = {Taylor \& Francis},
doi = {10.1080/00029890.1962.11989827},
}

@book{oxley2011matroid,
  title={Matroid Theory},
  author={James Oxley},
  series={Oxford Graduate Texts in Mathematics},
  year={2011},
  publisher={Oxford University Press}
}

@book{cook2011combinatorial,
  title={Combinatorial Optimization},
  author={Cook, W.J. and Cunningham, W.H. and Pulleyblank, W.R. and Schrijver, A.},
  series={Wiley Series in Discrete Mathematics and Optimization},
  year={1997},
  publisher={Wiley}
}

@article{FEIGENBAUM2005207,
title = {On graph problems in a semi-streaming model},
journal = {Theoretical Computer Science},
volume = {348},
number = {2},
pages = {207-216},
year = {2005},
doi = {https://doi.org/10.1016/j.tcs.2005.09.013},
author = {Joan Feigenbaum and Sampath Kannan and Andrew McGregor and Siddharth Suri and Jian Zhang},
}

@article{Muthukrishnan2005,
    author = {Muthukrishnan, S.},
    title = {Data Streams: Algorithms and Applications},
    journal = {Foundations and Trends in Theoretical Computer Science},
    volume = {1},
    number = {2},
    pages = {117--236},
    year = {2005},
    doi = {10.1561/0400000002},
}

@article{CKTY25,
author  = {Gergely Cs\'{a}ji and Tam\'{a}s Kir\'{a}ly and Kenjiro Takazawa and Yu Yokoi},
title   = {Popular maximum-utility matchings with matroid constraints},
journal = {Mathematics of Operations Research},
year    = {2025},
doi     = {10.1287/moor.2024.0633},
note    = {Published Online},
}

@misc{CKY24,
      title={Maximum Stable Matching with Matroids and Partial Orders}, 
      author={Gergely Cs\'{a}ji and Tam\'{a}s Kir\'{a}ly and Yu Yokoi},
      year={2024},
      note={arXiv preprint arXiv:2208.09583v3},
      url={https://arxiv.org/abs/2208.09583}, 
}

\end{document}